\documentclass[onefignum,onetabnum]{siamart190516}

\theoremstyle{plain}
\newtheorem{assumption}{Assumption}
\newtheorem{note}{Note}
\newtheorem{example}{Example}

\usepackage{tikz}
\usetikzlibrary{arrows,automata}

\usepackage{lipsum}
\usepackage{amsfonts}
\usepackage{graphicx}
\usepackage{epstopdf}
\usepackage{algorithmic}
\ifpdf
  \DeclareGraphicsExtensions{.eps,.pdf,.png,.jpg}
\else
  \DeclareGraphicsExtensions{.eps}
\fi

\newsiamremark{remark}{Remark}
\newsiamremark{hypothesis}{Hypothesis}
\crefname{hypothesis}{Hypothesis}{Hypotheses}
\newsiamthm{claim}{Claim}

\headers{Qualitative Behavior of a Hybrid Feedback Metabolic Pathway}{Claudia Lopez-Zazueta, and Vincent Fromion}

\title{Qualitative Behavior of a Metabolic Pathway with Hybrid Feedback\thanks{This article has been published in SIAM Journal on Applied Dynamical Systems (SIADS), \url{https://doi.org/10.1137/21M1451282}
\funding{This work was funded by Labex Mathématique Hadamard (LMH).}}}

\author{Claudia Lopez-Zazueta\thanks{Universite Paris-Saclay, INRAE, MaIAGE, Jouy-en-Josas, France 
  (\email{claudia.lopez-zazueta@inrae.fr}).}
\and Vincent Fromion \footnotemark[2]\thanks{Universite Paris-Saclay, INRAE, MaIAGE, Jouy-en-Josas, France 
  (\email{vincent.fromion@inrae.fr}).}}

\usepackage{amsopn}

\ifpdf
\hypersetup{
  pdftitle={Qualitative Behavior of a Metabolic Pathway with Hybrid Feedback 
  },
  pdfauthor={Claudia Lopez-Zazueta  and Vincent Fromion }}
\fi

\begin{document}

\maketitle

\begin{abstract}
We study the qualitative behavior of a model to represent local regulation in a metabolic network. The model is based on the end-product control structure introduced in [A. Goelzer, F. Bekkal Brikci, I. Martin-Verstraete et al., \emph{BMC Syst Biol} 2 (2008), pp. 1--18]
. In this class of regulation, the metabolite effector is the end-product of a metabolic pathway. We suppose the input to the pathway to switch between zero and a positive value according to the concentration of the metabolite effector. Considering the {\color{black}switching} system as a differential inclusion, we prove that it converges to a globally uniformly asymptotically stable equilibrium point, reaches the sliding mode or oscillates around the sliding mode depending on the positive value of the input. Finally, we show that in any case the solution of the {\color{black}switching} system is the limit of solutions of equation sequences with smooth or piecewise linear inputs. 
\end{abstract}



\section{Introduction}\label{introduction}

Metabolic networks are an important part of cells and understanding how they operate is an important issue whether in the context of human health or biotechnology. A significant number of methods for analyzing metabolic networks focus on the analysis of their equilibrium regimes.  Among these methods, two emblematic and well-known in the study of metabolic networks are Flux Balance Analysis (FBA) \cite{orth2010flux} (see also Metabolic Flux Analysis \cite{varma1994metabolic}) and Metabolic Control Analysis (MCA) \cite{heinrich1975mathematical, kacser1973control, rapoport1974linear,  reder1988metabolic}.

From the point of view of dynamical systems, these methods of metabolic network analysis make the implicit assumption that the metabolic network not only has an unique equilibrium regime, but that it is stable. Given the high predictive power of some methods based on this assumption, it can be considered that it is empirically validated at least at the level of cell populations. 

While this assumption of the existence of equilibrium regimes is undoubtedly fundamental in a large number of methods, and very useful in practice, the identification of the conditions that guarantee it is nonetheless essential. However, beyond the theoretical aspect, the conditions that ensure its validity determine our ability to intervene on metabolic networks, whether in therapeutic perspectives or in the context of biotechnologies. Indeed, it is {\color{black}important} to know whether the hypothesis of quasi-stationarity is preserved when, for example, certain enzymes are inhibited by drugs in the case of medical treatments or when a new pathway is added to the network in the context of biotechnologies and synthetic biology.

More fundamentally, knowing under which assumptions the metabolic network satisfies the quasi-stationarity assumption is a determining element in the analysis of the genetic regulation of metabolic networks. The question is vast and has already been addressed many times in the literature under different assumptions \cite{baldazzi2012importance, kuntz2014model, waldherr2015dynamic}. Modeling metabolic networks coupled with gene expression has been a subject of active research during the last decade \cite{baldazzi2012importance, goelzer2011cell, kuntz2014model, lerman2012silico, liu2020regulatory, waldherr2015dynamic}. Yet {\color{black}persistent} problems in metabolic modeling are large scale of models \cite{anderson2011model, gerdtzen2004non, lopez2018analytical, lopez2019dynamical, radulescu2008robust}, nonlinear kinetics \cite{horn1972general, savageau1969biochemical, voit2015150} and stochasticity \cite{kim2014validity, rao2003stochastic}.

Time-scale separation and the Quasi Steady State Assumption (QSSA) have been proposed as useful approaches to reduce deterministic models of metabolic networks \cite{goeke2014constructive, lopez2018analytical, lopez2019dynamical}, as well as for stochastic models of biochemical reactions and genetic networks \cite{el2005stochastic, kim2014validity, rao2003stochastic}. 

Also, the reduction through time-scale separation and QSSA has been applied to deterministic models of metabolic-genetic networks \cite{baldazzi2012importance, kuntz2014model, waldherr2015dynamic}. The method consists of dividing the states in two groups: the fast species (metabolites) and the slow species (macromolecules, gene products). Then, a deterministic model can be reduced using techniques for singularly perturbed systems (e.g. the theorems of Tikhonov \cite{khalil2002differential, kokotovic1999singular, tikhonov1985differential} and Fenichel \cite{fenichel1979geometric, verhulst2007singular}). The solution of the reduced system approximates the solution of the original system if some conditions are satisfied. One of these is asymptotic stability for the fast part of the system when the slow species are assumed to be constant.

In particular, the question of the stability of metabolic pathways with negative feedback loops has been considered in the past \cite{allwright1977global, arcak2006diagonal, chaves2019dynamics, mees1978periodic, meslem2011lyapunov, meslem2010stability, tyson1978dynamics, wang2010conditions}. This leads some authors to characterize the stability properties of linear metabolic pathways transforming an initial substrate into a final product of interest through $n$ elementary enzymatic reactions and where the concentration of the last metabolite, i.e. the end-product, negatively modulates the activity of the first enzyme. In this context, the authors mainly tried to identify the conditions that ensure the stability (in the Lyapunov sense) of such linear pathways with negative feedback. The first studies have mainly investigated the stability properties of the linearization of the system associated with its equilibrium point. {\color{black} For example, Tyson and Othmer in \cite{tyson1978dynamics} studied the stability of a negative feedback system with linear and irreversible kinetics using a secant criterion. Arcak and Sontag in \cite{arcak2006diagonal} extended these results for nonlinear negative feedback systems with irreversible kinetics where, in connection to the small-gain theorem, a secant criterion equivalent to diagonal stability was also used.}

{\color{black} 
On the other hand, other works have addressed the oscillatory behavior of smooth negative feedback loops. For instance, Tyson and Othmer in \cite{tyson1978dynamics} have proved the existence of oscillatory and periodic solutions for a linear negative feedback system with irreversible reactions. 
Hasting et al. in \cite{hastings1977existence} have given a geometrical proof of the existence of a non-constant periodic solution for a continuous system of ordinary differential equations ({\color{black}ODEs}) of class $C^n$. The system represents a negative feedback loop with monotone and non-reversible reaction kinetics and the condition for the oscillations is given in terms of the eigenvalues of the Jacobian matrix at an equilibrium point. Mallet-Paret and Smith have proved in \cite{mallet1990poincare} the Poincare-Bendixson theorem for a monotone feedback loop of class $C^1$ satisfying a convexity condition. As well as in \cite{hastings1977existence}, their feedback loop does not account for reversible kinetics. Poignard et al. in \cite{poignard2016periodic} also consider a negative feedback loop but with decay rates and ODEs that are monotonic except in a narrow window around a threshold value. The existence of a periodic orbit for this system is proved circumscribing it by two piecewise linear systems. Then in \cite{poignard2018stability}, its uniqueness and asymptotic stability are proved under some symmetry assumptions on the parameters and assuming that all decay rates are equal.

Piecewise linear equations have been also proposed to model oscillations in biological control systems. Glass and Pasternack in \cite{glass1978stable} have given conditions to prove the existence of a stable limit cycle for a piecewise linear differential equation. These conditions have to be verified in the state transition diagram defined for piecewise linear systems and an algebraic computation allows to determine the existence of the stable limit cycle. Farcot and Gouzé in \cite{farcot2009periodic} have studied a piecewise linear equation describing a negative feedback loop with irreversible kinetics and non-identical decay rates. Using a fix point theorem, they have proved the existence and uniqueness of a stable periodic orbit in dimension 3 or more. Using a formulation with piecewise constant matrices, Quee and Edwards in \cite{quee2021ramp} have proved the existence of a non-constant periodic solution for piecewise affine system representing a negative feedback loop with non-reversible kinetics and non-identical decay rates. }

The purpose of this work is to present a generic model that represents allosteric regulation for a repressible enzyme, i.e., allosteric inhibition. For a repressible enzyme, the presence of the effector molecule enhances the binding of a repressor molecule to the operator gene that regulates the enzyme coding, and transcription is blocked \cite{jacob1961genetic, tyson1978dynamics}. The fraction of the operator region free of repressor corresponds to a monotone decreasing function with respect to the metabolite effector concentration. Moreover, these reactions occur quickly and are therefore at equilibrium \cite{tyson1978dynamics}. In the limit case, we can consider that the monotone decreasing function is a step function (see \Cref{repress}).

\begin{figure}[!]\centering
\includegraphics[width=7cm]{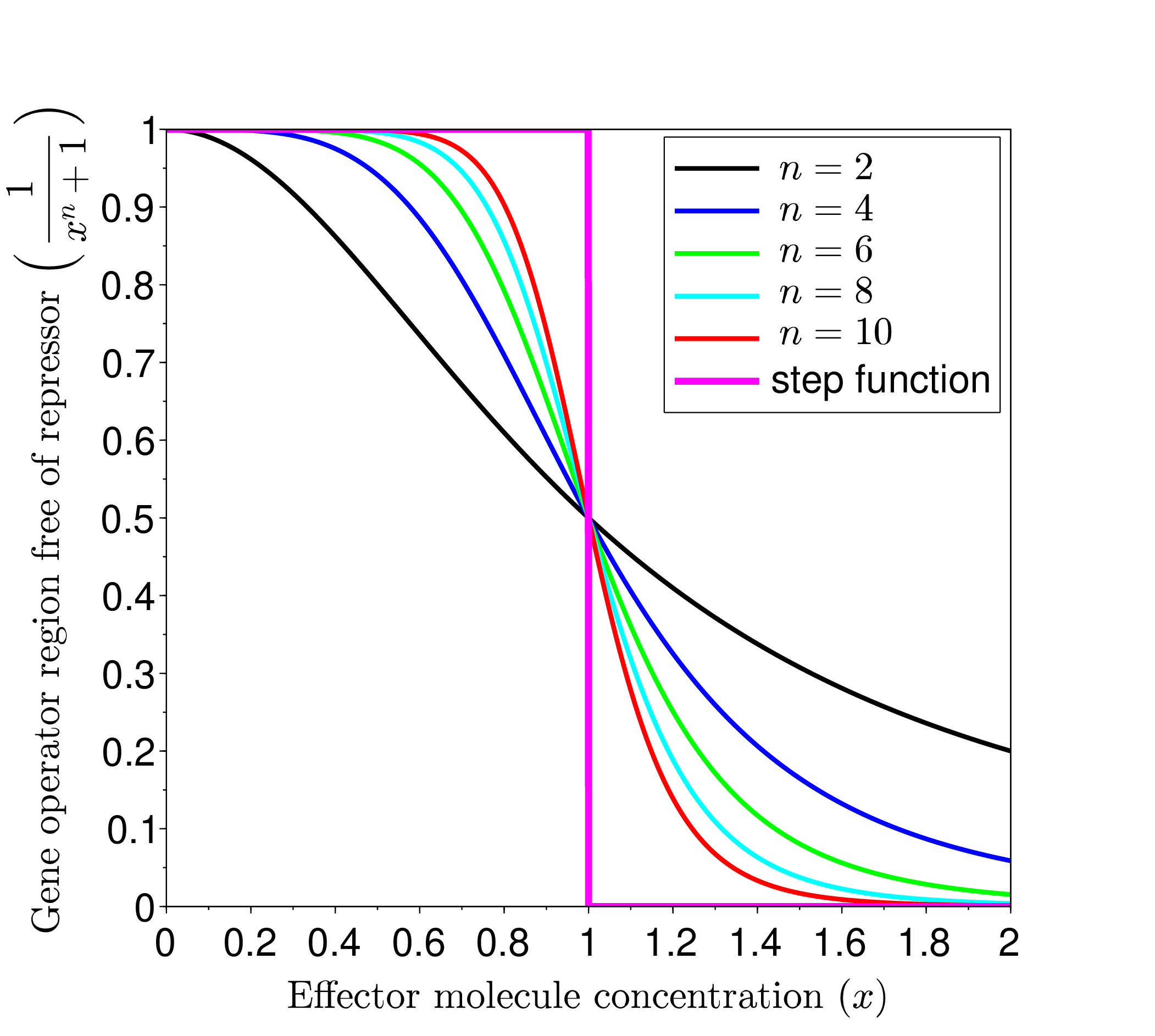}\\
\caption{In the case of a repressible enzyme, the fraction of the operator region free of repressor correspond to a monotone decreasing function. The reactions that block transcription occur fast, which in a limit case can be represented as a step function.}\label{repress}
\end{figure} 

For this purpose, we introduce a model that represents local regulation in a metabolic pathway and study its dynamical behavior. The model is based on the end-product control structure introduced in \cite{goelzer2008reconstruction} (see also \cite{goelzer2014towards}). In this class of local regulation, the metabolite effector is the end product of the pathway and enzyme synthesis of the pathway is induced when the concentration of the effector decreases. In a first stage, we consider that enzyme concentrations remain constant with the purpose of using a slow-fast system approach. 

Indeed, the experiments indicate that quasi-steady state exist{\color{black}s} at the scale of a cell population. However, to assess the stability of metabolic systems is challenging due to the large scale of models and their non-linear dynamics. Even if it is classic in the engineering context to consider that stability is an expected or even necessary property, other behaviors of the system may also be acceptable from the standpoint of the functioning of a dynamic system, as for example the fact that the system oscillates about an average value. Besides, multistability and oscillatory dynamics can emerge in metabolic pathways under gene regulation \cite{el2005stochastic, fung2005synthetic, wang2010conditions, oyarzun2012multistability, chaves2019dynamics}. 

In the context of cells, where the stochasticity is very large, fluctuations due to these oscillations would not necessarily be a problem.  In this paper, we thus take the counterpart of the classic approach by considering that the negative feedback is high. We will obtain {\color{black}conditions ensuring the existence of an oscillatory regime} by considering a limit case, i.e. by assuming that the feedback is an ON/OFF type mechanism.

In \Cref{sec_model}, we describe the model of local regulation. In order to represent the stiffness of allosteric regulation, we consider a deterministic model of a linear metabolic pathway with an input-feedback that switches between two modes (ON and OFF) according to the concentration of the metabolite effector. The fluxes among the metabolites are generalized so that kinetics can be nonlinear, only respecting some monotonicity conditions.

In \Cref{sec_hybrid}, we define a differential inclusion for the {\color{black}switching} system following the theory of Filippov for discontinuous systems \cite{filippov1988differential}. We prove that, under some conditions, the solution of the differential inclusion can converge uniformly and asymptotically to an equilibrium point, remain in equilibrium at the sliding mode or oscillate around the sliding mode. 

In \Cref{section_rMM}, we present an example of a metabolic pathway with Michaelis-Menten reversible reactions. Finally, in \Cref{sec_continuous}, we show that the solution of the {\color{black}switching} system is the limit of function sequences with smooth or piecewise linear inputs that tend towards an ON / OFF type mechanism.

\section{Model of local regulation}\label{sec_model}

In this Section we introduce the feedback model studied through the text, which is based on the end-product control structure proposed in \cite{goelzer2008reconstruction} (see also \cite{goelzer2014towards}). The model corresponds to a metabolic pathway where the end product is the metabolite effector. 

We consider the model as a slow-fast system, where metabolites are the fast species and enzymes the slow species. In order to analyze the stability conditions for the fast part of the system according to Tikhonov's Theorem \cite{khalil2002differential, kokotovic1999singular, tikhonov1985differential, verhulst2007singular}, we assume that enzyme concentrations are constant. 

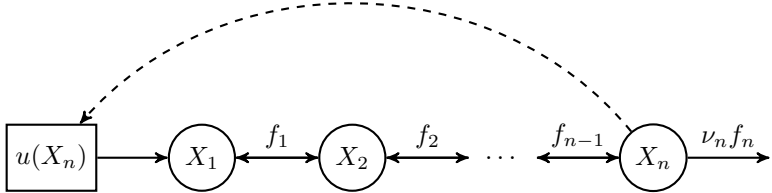
\begin{figure}[!]
\centering
\begin{tikzpicture}[->,>=stealth',shorten >=0pt,auto,node distance=2cm,thick]
\tikzstyle{every state}=[rectangle, fill=white,draw=black,text=black]
 	\node[state]		(u)							{$u(X_n)$};                   
\tikzstyle{every state}=[fill=white,draw=black,text=black]
 	\node[state]		(A)		[right of=u]			{$X_1$};                            
  	\node[state]		(B)		[right of=A]			{$X_2$};
  	\tikzstyle{every state}=[fill=white,draw=white,text=black]
 	\node[state]		(C)		[right of=B]			{$\dots$};
 	\tikzstyle{every state}=[fill=white,draw=black,text=black]
  	\node[state]		(D)		[right of=C]			{$X_n$};
\tikzstyle{every state}=[rectangle, fill=white,draw=white,text=white]
 	\node[state]		(out)			[right of=D]				{};       
\path
(u)			edge			[]			node	{}											(A)
(A)			edge			[]			node	{$f_1$}								(B)
(B)			edge			[]			node	{}											(A)
(B)			edge			[]			node	{$f_2$}								(C)
(C)			edge			[]			node	{}											(B)
(C)			edge			[]			node {$f_{n-1}$}						(D)
(D)			edge			[]			node {}											(C)
(D)			edge			[dashed, bend right=50*+1]		node {}					(u)
(D)			edge			[]			node {$\nu_nf_n$}										(out)	;
\end{tikzpicture}
\caption{End-product control structure.}\label{meta_path}
\end{figure} 

The ODE describing the concentration of metabolites in the pathway of \Cref{meta_path} is
\begin{align}\label{ode_meta_path}
\frac{dX_1}{dt}&=u(X_n)-f_1(X_1,X_2)-\mu\cdot X_1\\\notag
\frac{dX_2}{dt}&=f_1(X_1,X_2)-f_2(X_2,X_3)-\mu\cdot X_2\\\notag
\vdots\\\notag
\frac{dX_n}{dt}&=f_{n-1}(X_{n-1},X_n)-\nu_nf_n(X_n)-\mu\cdot X_n\notag,
\end{align}where $u\colon\mathbb R^n\to[0,\infty)$ is the input function, $\mu\geq0$ the growth rate {\color{black}(or decay rate)} and $\nu_n\geq0$ an output rate. In this case we consider that allosteric regulation acts on the input flux of the pathway, which can be interpreted as the regulation of the enzyme activity leading to the input by the metabolite effector $X_n$ (the end-product of the pathway).

In order to simplify the notation we define
\begin{align*}
f(u,X,\mu,\nu_n):=\begin{pmatrix}
u(X_n)-f_1(X_1,X_2)-\mu\cdot X_1\\
f_1(X_1,X_2)-f_2(X_2,X_3)-\mu\cdot X_2\\
\vdots\\
f_{n-1}(X_{n-1},X_n)-\nu_nf_n(X_n)-\mu\cdot X_n
\end{pmatrix}.
\end{align*}Then, \Cref{ode_meta_path} can be rewritten as
\begin{align*}
\frac{dX}{dt}=f(u,X,\mu,\nu_n).
\end{align*}

We assume that kinetics of the metabolic pathway in \Cref{meta_path} can be nonlinear, but they respect some monotonicity conditions that are established in Assumption \ref{assumption1}. The monotonicity condition implying that functions $f_i$ are strictly increasing with respect to the first entry assures the flux from the input to the end-product of the pathway. On the other hand, we consider that reactions can be reversible, but this is not imposed as a condition. Therefore, functions $f_i$ are supposed to be decreasing (but not strictly) with respect to the second entry.

\begin{assumption}\label{assumption1}
For every $i=1,\dots,n-1$ assume 
\begin{itemize}
\item[(i)] $f_i:[0,\infty)\times[0,\infty)\to\mathbb R$ and $f_n(X_n):[0,\infty)\to\mathbb R$ are continuous functions,
\item[(ii)]$f_i(X_i,X_{i+1})$ is strictly increasing w.r.t. $X_i$ and $f_n(X_n)$ is strictly increasing w.r.t. $X_n$, i.e. 
\begin{align*}
f_i(X_i,X_{i+1})&<f_i(X_i',X_{i+1})&\forall X_i<X_i',\\
&& \forall X_{i+1}\in [0,\infty),\\
f_n(X_n)&<f_n(X_n')&\forall X_n<X_n',
\end{align*}
\item[(iii)] $f_i(X_i,X_{i+1})$ is decreasing w.r.t. $X_{i+1}$, i.e.
\begin{align*}
f_i(X_i,X_{i+1})&\leq f_i(X_i,X'_{i+1})&\forall X'_{i+1}<X_{i+1},\\
&& \forall X_i\in [0,\infty),
\end{align*}
\item[(iv)] $f_i(0,X_{i+1})\leq0$ for all $X_{i+1}\in [0,\infty)$, 
\item[(v)]$f_i(X_i,0)\geq0$ for all $X_i\in [0,\infty)$ and
\item[(vi)] and $f_n(0)=0$.
\end{itemize}
\end{assumption}

\section{Hybrid model}\label{sec_hybrid}

We consider the limit case where allosteric regulation can be represented by a step function, which can be understood as the induction or repression of the enzyme activity (see \Cref{repress}). For this purpose, we assume the input of system \eqref{ode_meta_path} to be a step function that depends on the concentration of $X_n$ the metabolite effector with respect to $K>0$ a threshold value:
\begin{align}\label{input}
u[X_n]:=\begin{cases}
k_1 & \text{ if } X_n< K \\
0 & \text{ if } K<X_n
\end{cases}.
\end{align}Through the text, we refer to $k_1$ as the constant input ON and $0$ as the constant input OFF.

Hence, we suppose that the enzyme activity leading to the input is induced if the concentration of the metabolite effector is under a threshold $K$ (i.e. $X_n<K$), and it is repressed when the concentration of the metabolite effector is over the threshold (i.e. $K<X_n$).

\subsection{Oscillatory system}

Switching the input allows to keep the concentration of the metabolite effector $X_n$ as close as possible to the threshold $K$. Indeed, the input is OFF when the metabolite effector concentration exceeds the threshold, which allows a decrease of the flux pathway and of $X_n$ concentration consequently. Reciprocally, the input is ON if the metabolite effector concentration is under the threshold, which leads to an increment of the flux pathway and $X_n$ concentration.  

The next \Cref{proposition1} states that, if the constant input ON is enough large and the sliding mode is not attained, then the switching leads to an oscillatory behavior in all the states of the system. {\color{black}In \cite{wang2010conditions}, the same characteristic was observed for linear monotone tridiagonal systems with nonlinear negative feedback. For the example presented in equation (21) of \cite{wang2010conditions}, they have shown that the equilibrium of the system does not satisfy their stability conditions when the effect of the input is lower than a threshold. Then, numerical simulations have put in evidence oscillations in that system with a large value input.}

\begin{theorem}\label{proposition1}
Under Assumption \ref{assumption1}, consider the system
\begin{align}\label{end_product_feedback}
\frac{dX}{dt}&=f(u[X_n],X,\mu,\nu_n)
\end{align}with initial conditions $X_i(t_0)\geq0$, $i=1,2,\dots,n$, $K>0$, $k_1>0$, $\mu\geq0$, $\nu_n\geq0$, $\mu+\nu_n>0$ and the input $u[X_n]$ defined in \eqref{input}. 

Suppose that there exist positive values $Y_1^*,Y_2^*,\dots,Y_{n-1}^*$ such that 
\begin{align*}
0&=f_1(Y_1^*,Y_2^*)-f_2(Y_2^*,Y_3^*)-\mu\cdot Y_2^*\\\notag
0&=f_2(Y_2^*,Y_3^*)-f_3(Y_3^*,Y_4^*)-\mu\cdot Y_3^*\\\notag
\vdots\\\notag
0&=f_{n-1}(Y_{n-1}^*,K)-\nu_nf_n(K)-\mu\cdot K\notag
\end{align*}
and that 
\begin{align}\label{hyp_alpha}
\sum_{i=1}^{n-1}\mu\cdot Y_i^* + \mu\cdot K + \nu_nf_n(K)<k_1.
\end{align}

Then, there is an absolutely continuous function $X=(X_1,X_2,\dots,X_n)$ that satisfies equation \eqref{end_product_feedback} for a.e. $t\in[t_0,\infty)$ and right uniqueness holds in $[t_0,\infty)\times[0,\infty)^n$. 

Furthermore, if there is $t^*\geq t_0$ such that $X_i(t^*)=Y_i^*$ for every $i=1,2,\dots,n-1$ and $X_n(t^*)=K$, then the {\color{black}switching} system \eqref{end_product_feedback} remains at sliding mode, i.e., for all $t\geq t^*$, $X_i(t)=Y_i^*$ for every $i=1,2,\dots,n-1$ and $X_n(t)=K$.

Otherwise, the {\color{black}switching} system \eqref{end_product_feedback} oscillates around $(Y_1^*,Y_2^*,\dots,Y_{n-1}^*,K)$. In other words, $X_i$ has an oscillatory trajectory that takes the value $Y_i^*$ infinitely many times for every $i=1,2,\dots,n-1$ and $X_n$ has an oscillatory trajectory that takes the value $K$ infinitely many times.
\end{theorem}

The proof of \Cref{proposition1} is in \Cref{sec:existence}. In the next subsection we present some results necessary to it. 

{\color{black}
\begin{note}
To explain the key condition \eqref{hyp_alpha}, suppose that there exists a constant value $\alpha\cdot k_1$ such that the system with constant input
$$\frac{dX}{dt}=f(\alpha\cdot k_1,X,\mu,\nu_n)$$
has an equilibrium point $(Y_1^*, Y_2^*,\dots, Y_n^*)$ with $Y_n^*=K$. Then, by the mass-action kinetics, we have
$$\alpha\cdot k_1=\sum_{i=1}^{n-1}\mu\cdot Y_i^*+\mu\cdot K+\nu_nf_n(K).$$This implies by condition \eqref{hyp_alpha} that
$$\alpha\cdot k_1<k_1.$$In other words, that the input necessary to stabilize $X_n$ at $K$ is less than $k_1$. In this sense, to have a ``large'' input $k_1$ implies that the {\color{black}switching} system oscillates around $(Y_1^*, Y_2^*,\dots, Y_{n-1}^*,K)$, as an alternative to the stabilization.
\end{note}

\begin{note}
The sliding mode equilibrium has also been called singular equilibrium in the literature \cite{casey2006piecewise, duncan2021equilibria, mestl1995mathematical, plahte1994global}. In \cite{casey2006piecewise}, an approach to model genetic regulatory networks has been proposed using equations described by piecewise linear functions and differential inclusions. It differs from our approach since we consider hybrid systems with ODEs described by continuous {\color{black}monotonic} functions. Moreover, in \cite{casey2006piecewise} it is only given the characterization for the stability of singular equilibria, while in this paper the existence of oscillatory regimes is also rigorously proved, which is an important feature for metabolic pathways under genetic regulation. In \cite{duncan2021equilibria}, the stability of cyclic feedback networks with sigmoidal and irreversible kinetics has been studied by means of their Jacobian matrices. Feedback systems are not analyzed in \cite{mestl1995mathematical, plahte1994global}, but results on the relationship of stable and periodic solutions of logoid systems and piecewise linear equations are presented in \cite{plahte1994global} and conditions necessary for the stability of steady states are obtained from a Logoid-Jacobian matrix for gene regulatory networks in \cite{mestl1995mathematical}.
\end{note}}

\subsection{Systems with constant inputs}
To prove \Cref{proposition1}, it is useful to analyze the metabolic pathway system \eqref{ode_meta_path} when the input is a constant function. In this section, we introduce several lemmas for systems with constant inputs that are used in the proof of \Cref{proposition1}. The proofs of all Lemmas are in \Cref{supp_mat}.

\begin{definition}
We say that a vector $X=(X_1,X_2,\dots,X_n)\in\mathbb R^n$ is nonnegative if $X_i\geq0$ for all $i=1,2,\dots,n.$

{\color{black}Similarly}, a vector $X=(X_1,X_2,\dots,X_n)\in\mathbb R^n$ is positive if $X_i>0$ for all $i=1,2,\dots,n.$
\end{definition}

The next lemma states that, in case of having a nonnegative constant input, system \eqref{ode_meta_path} is positively invariant. Moreover, if the system has a nonnegative equilibrium point, this delimits the boundary of some invariant regions.

\begin{lemma}\label{lemma4}

Under Assumption \ref{assumption1}, consider the system
\begin{align*}
\frac{dX}{dt}=f(\mathbf I,X,\mu,\nu_n),
\end{align*}where $\mathbf I$ indicates a constant input function with value $\mathbf I\geq0$, $\nu_n\geq0$ and $\mu\geq0$. Then,
\begin{align*}
\Omega_1&:=[0,\infty)\times[0,\infty)\times\dots\times[0,\infty)
\end{align*}is positively invariant under the flux $X$.

Moreover, suppose that the system above has a nonnegative equilibrium point $X^*:=(X_1^*,X_2^*,\dots,X_n^*)$. Then, the subsets
\begin{align*}
\Omega_2&:=[0,X_1^*]\times[0,X_2^*]\times\dots\times[0,X_n^*],\\
\Omega_3&:=[X_1^*,\infty)\times[X_2^*,\infty)\times\dots\times[X_n^*,\infty),\\
\Omega_4&:=\{X_1^*\}\times\{X_2^*\}\times\dots\times\{X_n^*\},
\end{align*}are positively invariant under the flux $X$. 
\end{lemma}

The next proposition claims that if system \eqref{ode_meta_path}  with a nonnegative constant input has a nonnegative equilibrium point, then, this is globally uniformly asymptotically stable. The proof is divided in two cases. In the case when $\mu>0$, the proof consists on defining a Lyapunov function that is bounded by a positive definite function. Then, using an extension of the LaSalle invariance principle \cite{orlov2009discontinuous}, the result is concluded. 

In the other case, when $\mu=0$ and $\nu_n>0$, the proof follows the ideas of the particular case with Michaelis-Menten kinetics presented in Proposition 8 of \cite{ndiaye2013global}. Using that the Jacobian is a column diagonally dominant matrix due to the monotonicity conditions of Assumption \ref{assumption1}, it is proved that the equilibrium point is globally attractive and also locally asymptotically stable (see Appendix \ref{proof_gouze}). 

\begin{proposition}\label{lemma1}
Under Assumption \ref{assumption1}, consider the system
\begin{align}\label{input_on_I}
\frac{dX}{dt}=f(\mathbf I,X,\mu,\nu_n),
\end{align}with initial conditions $X_i(t_0)\geq 0$ for every $i=1,2,\dots,n$, the constant input $\mathbf I\geq0$, $\mu\geq0$, $\nu_n\geq0$ and $0<\mu+\nu_n$. 

If system \eqref{input_on_I} has a nonnegative equilibrium point $X^*:=(X_1^*,X_2^*,\dots,X_n^*)$, then $X^*$ is globally uniformly asymptotically stable (GUAS).
\end{proposition}

\begin{proof}{
First suppose that $0<\mu$. Define the Lyapunov norm-like function $$V(X):=\sum_{i=1}^n\vert X_i-X_i^*\vert,$$ where $X:=(X_1, X_2,\dots X_n)$. The function $V$ is continuous, nonnegative and $V(X)=0$ if and only if $X=X^*$ (i.e. $V$ is positive definite).

According to \Cref{lemma4}, $X_i(t)\in[0,\infty)$ for every $t\geq t_0$, $i=1,2,\dots,n$. Then, as a consequence of the monotonicity of the functions $f_i$ established in Assumption \ref{assumption1} and the existence of the equilibrium point, it follows,

{\color{black}
\begin{align*}
\dot V(X)&=\sum_{i=1}^{n-1}f_i(X_i,X_{i+1})[sgn(X_{i+1}-X_{i+1}^*)-sgn(X_i-X_i^*)]\\
&+\mathbf I\cdot sgn(X_1-X_1^*)-\nu_nf_n(X_n)\cdot sgn(X_n-X_n^*) - \sum_{i=1}^n \mu\cdot X_i\cdot sgn(X_i-X_i^*)\\
&\leq -\sum_{i=1}^{n-1}\sum_{j=1}^i \mu\cdot X_j^*[sgn(X_{i+1}-X_{i+1}^*)-sgn(X_i-X_i^*)]\\\notag
&+ \sum_{i=1}^n \mu\cdot X_i^*sgn(X_n-X_n^*)  - \sum_{i=1}^n \mu\cdot X_i\cdot sgn(X_i-X_i^*)\\
&\leq\sum_{i=1}^{n}\mu\cdot (X_i^*-X_i)\cdot sgn(X_i-X_i^*)=-\sum_{i=1}^n \mu\cdot\vert X_i-X_i^*\vert.
\end{align*}}

Defining $W(X):=\sum_{i=1}^n \mu\cdot\vert X_i-X_i^*\vert$, we have that $W(X)$ is a continuous nonnegative function such that $W(X)=0$ if and only if $X=X^*$ and
\begin{align*}
\dot V(X)&\leq -W(X). 
\end{align*}

Therefore, by means of an extension of LaSalle invariance principle (see Theorem 3.3 in \cite{orlov2009discontinuous}), we conclude that $X^*$ is globally uniformly asymptotically stable.

The proof for the case $\mu=0$ and $0<\nu_n$ is in Appendix \ref{proof_gouze}.
}\end{proof}

The next \Cref{lemma5} states an order for the metabolic pathways with constant inputs. In other words, it compares two systems of the form \eqref{ode_meta_path} with constant inputs, according to the input values and their initial conditions.

Moreover, \Cref{lemma2} asserts an order for nonnegative equilibrium points of two systems of the form \eqref{ode_meta_path} with positive constant inputs.
 
Finally, \Cref{lemma9} states that if a system of the form \eqref{ode_meta_path} with constant input has a nonnegative equilibrium point, then there is also a system of the form \eqref{ode_meta_path} with a larger constant input that has a larger (entry by entry) equilibrium point. 

\begin{lemma}\label{lemma5}
Under Assumption \ref{assumption1}, consider two systems
\begin{align*}
\frac{dX}{dt}&=f(\mathbf I_1,X,\mu,\nu_n)\\
\frac{dY}{dt}&=f(\mathbf I_2,Y,\mu,\nu_n),
\end{align*}with constant inputs that satisfy 
$$0\leq\mathbf I_2\leq\mathbf I_1,$$$\nu_n\geq0$, $\mu\geq0$ and initial conditions $X(t_0)$ and $Y(t_0)$, respectively, such that
\begin{align*}
0\leq Y_i(t_0)&\leq X_i(t_0)&\forall i=1,2,\dots,n.
\end{align*}If $\mathbf I_2<\mathbf I_1$, then
\begin{align*}
Y_i(t)&< X_i(t)&\forall t\in(t_0,\infty), \forall i=1,2,\dots,n.
\end{align*}
Moreover, if $\mathbf I_1=\mathbf I_2$ and $Y_i(t_0)<X_i(t_0)$ for every $i=1,2, \dots,n$, there is $T>t_0$ such that
\begin{align*}
Y_i(t)&< X_i(t)&\forall t\in(t_0,T), \forall i=1,2,\dots,n,
\end{align*}and
\begin{align*}
Y_i(t)&\leq X_i(t)&\forall t\in[T,\infty), \forall i=1,2,\dots,n.
\end{align*}
\end{lemma}

\begin{lemma}\label{lemma2}
Under Assumption \ref{assumption1}, consider the following systems 
\begin{align}\label{input_on_I1}
\frac{dX}{dt}&=f(\mathbf I_1,X,\mu,\nu_n),\\
\label{input_on_I2}
\frac{dY}{dt}&=f(\mathbf I_2,Y,\mu,\nu_n),
\end{align}where $\nu_n\geq0$ and $\mu\geq0$ and the constant inputs satisfy $$0\leq\mathbf I_2<\mathbf I_1.$$

Suppose that systems \eqref{input_on_I1} and \eqref{input_on_I2} have nonnegative equilibrium points $X^*:=(X_1^*,X_2^*,\dots,X_n^*)$ and $Y^*:=(Y_1^*,Y_2^*,\dots,Y_n^*)$, respectively. Then,
\begin{align*}
 Y_i^*&<X_i^*&\forall i=1,2,\dots,n.
 \end{align*} \end{lemma}

\begin{lemma}\label{lemma9}
Under Assumption \ref{assumption1}, suppose that there is a nonnegative vector $X^*:=(X_1^*,X_2^*,\dots,X_n^*)$ such that
\begin{align*}
f(\mathbf I_1,X^*,\mu,\nu_n)=\bar0,
\end{align*}with $\mathbf I_1>0$, $\nu_n\geq0$ and $\mu\geq0$. 

Then, for any $\varepsilon>0$, there exist
$$\mathbf I\in(\mathbf I_1,\mathbf I_1+\varepsilon)$$and a nonnegative vector $X':=(X_1',X_2',\dots,X_n')$ such that  
\begin{align*}
X_i^*&<X_i' &\forall i=1,\dots,n
\end{align*}and
\begin{align*}
f(\mathbf I,X',\mu,\nu_n)=\bar0.
\end{align*}
\end{lemma}

\subsection{Solution existence and right uniqueness}\label{sec:existence}

In this section we will prove the existence and right uniqueness of an absolutely continuous solution for the {\color{black}switching} system \eqref{end_product_feedback}. For this purpose, we use the theory of differential inclusions of Filippov \cite{filippov1988differential}. {\color{black}The proofs of all Lemmas are in \Cref{supp_mat}.}

\begin{definition}
We say that for the equation
\begin{align*}
\frac{dx}{dt} = f(t,x)
\end{align*}\emph{right uniqueness} holds at a point $(t_0, x_0)$ if there exists $t_1 > t_0$ such that each two solutions of this equation satisfying the condition $x(t_0) = x_0$ coincide on the interval $t_0 \leq t \leq t_1$ or on the part of this interval on which they are both defined.
Moreover, \emph{right uniqueness holds in a domain $D$} (open or closed) if for each point $(t_0, x_0)\in D$ every two solutions satisfying the condition $x(t_0) = x_0$ coincide on each interval $t_0\leq t \leq t_1$ on which they both exist and lie in this domain \cite{filippov1988differential}.
\end{definition}

\begin{definition}
We define the sign as a function $sgn\colon\mathbb R\to\{-1,0,1\}$ such that
$$sgn(x):=\begin{cases}
-1&\text{if }x<0\\
0&\text{if }x=0\\
1&\text{if }x>0
\end{cases}.$$
\end{definition}

The purpose of \Cref{lemma11} and \Cref{lemma12} is to analyze the behavior of the {\color{black}switching} system \eqref{end_product_feedback} when its last state (metabolite $X_n$) takes the value $K$ at which the systems switches. This unequivocally defines the value taken by the entry of the {\color{black}switching} system \eqref{end_product_feedback} in the differential inclusion, allowing to conclude in \Cref{lemma_uniqueness} that \emph{right uniqueness} holds for its solution.

\begin{lemma}\label{lemma11}
Under Assumption \ref{assumption1}, consider the system
\begin{align*}
\frac{d X}{dt}&=f(\mathbf I,X,\mu,\nu_n),
\end{align*}with $\mathbf I\geq0$, $\mu\geq0$ and $\nu_n\geq0$. 
Suppose that for some $m\in\{1,2,\dots,n-1\}$
\begin{align*}
\dot X_i(t_0)&=0&&\forall i>m.
\end{align*}Then, there exists $\varepsilon>0$ such that
\begin{align*}
sgn(\dot X_i(t))&=sgn(\dot X_m(t))&\forall t\in(t_0,t_0+\varepsilon),\forall i>m.
\end{align*}
\end{lemma}

\begin{lemma}\label{lemma12}
Under Assumption \ref{assumption1}, let $k_1>0$, $\mu\geq0$, $\nu_n\geq0$ and consider the systems
\begin{align*}
\frac{d\widetilde X}{dt}&=f(k_1,\widetilde X,\mu,\nu_n)\\
\frac{dZ}{dt}&=f(0,Z,\mu,\nu_n),
\end{align*}with the same initial conditions
\begin{align*}
\widetilde X_i(t_0)&=Z_i(t_0)&\forall i=1,2,\dots,n.
\end{align*}

Assume that there is $m\in\{2,3,\dots,n\}$ such that 
$$\dot{\widetilde X}_m(t_0)=\dot Z_m(t_0)\not=0.$$

Then, there exists $\varepsilon>0$ such that 
\begin{align*}
sgn(\dot{\widetilde X}_n(t))&=sgn(\dot Z_n(t))&\forall t\in(t_0,t_0+\varepsilon).
\end{align*}
\end{lemma}

\begin{proposition}[Existence and uniqueness]\label{lemma_uniqueness}
Under Assumption \ref{assumption1}, consider the system
\begin{align}\label{switch_sys}
\frac{dX}{dt}&=f(u[X_n],X,\mu,\nu_n)
\end{align}with initial conditions $X_i(t_0)\geq0$, $i=1,2,\dots,n$, $K>0$, $k_1>0$, $\mu\geq0$, $\nu_n\geq0$ and the input $u[X_n]$ defined in \eqref{input}. 

Then, there exists an absolutely continuous function $X$ that satisfies \eqref{switch_sys} for almost every (a.e.) $t\in[t_0,\infty)$ and right uniqueness holds in $[t_0,\infty)\times[0,\infty)^n$.
\end{proposition}
\begin{proof}{
Consider the differential inclusion
\begin{align*}
\frac{dX}{dt}\in\mathbf F(t,X)
\end{align*}with
\begin{align*}
\mathbf F(t,X):=\begin{cases}
\big\{
f(k_1,X,\mu,\nu_n)
\big\}&\text{if }X_n< K\\
\big\{
f(0,X,\mu,\nu_n)
\big\}&\text{if }K<X_n\\
\big\{
f(\alpha\cdot k_1,X,\mu,\nu_n)\colon \alpha\in[0,1]
\big\}&\text{if }X_n=K
\end{cases}.
\end{align*}
For every $(t,X)\in[0,\infty)\times\mathbb R^n_+$, $\mathbf F(t,X)$ satisfies the basic conditions \cite{filippov1988differential}: $\mathbf F(t,X)$ is nonempty, bounded, closed, convex and upper semi-continuous in $(t,X)$ according to Lemma 3, p. 67 of \cite{filippov1988differential}. Then, by Theorem 1 and Theorem 2, pp. 77-78 of \cite{filippov1988differential}, there exists an absolutely continuous function $X$ that satisfies \eqref{switch_sys} for a.e. $t\in[t_0,\infty)$. 

Right uniqueness follows from the fact that for every $t\in[t_0,\infty)$ the derivative of an absolutely continuous solution can only take a single value which is given in agreement with the differential inclusion. Indeed, the derivative is uniquely determined in a neighborhood of any $t$ such that $X_n(t)\not=K$. 

On the other hand, if $X_n(t')=K$ and there is $m\in\{2,3,\dots,n\}$ such that $\dot X_m(t')\not=0$, by \Cref{lemma12}, there is only one valid definition for $\dot X$ in $(t',t'+\varepsilon)$, because the options with input $k_1$ and $0$ has the same sign in $\dot X_n$ in $(t',t'+\varepsilon)$, and therefore, one of this will fail to the restriction regarding the value of $X_n$ with respect to $K$. And the option with an input $\alpha\cdot k_1$ can only be taken in a complete interval when the system has reached an equilibrium, which cannot be the case since we are supposing $\dot X_m(t')\not=0$.

Finally, if $X_n(t')=K$ and $\dot X_i(t')=0$ for every $i\in\{2,3,\dots,n\}$, then the system takes the value $Y^*$ satisfying $Y_n^*=K$, which corresponds to the equilibrium of a system with input $\alpha\cdot k_1$ for some $\alpha\in\mathbb R$. If $1\leq\alpha$, then $\dot X_1(t')\leq0$ according to \Cref{lemma2}. Then, by \Cref{lemma11}, $X_n(t)\leq K$ for $t\in(t',t'+\varepsilon)$ with any of the inputs $k_1$ or 0. Therefore, the system takes the mode with input $k_1$. If $\alpha\in(0,1)$, the options with inputs $k_1$ an{\color{black}d} 0 fail because
\begin{align*}
0&=\alpha\cdot k_1-f_1(Y^*_1,Y^*_2)-\mu\cdot Y^*_1\\
&< k_1-f_1(Y^*_1,Y^*_2)-\mu\cdot Y^*_1,
\end{align*}
\begin{align*}
-f_1(Y^*_1,Y^*_2)-\mu\cdot Y^*_1=-\alpha\cdot k_1<0,
\end{align*}which, according to \Cref{lemma12}, means for the system with input $k_1$ that $K< X_n(t)$ for $t\in(t',t'+\varepsilon)$ and for the system with input $0$ that $X_n(t)< K$ for $t\in(t',t'+\varepsilon)$, which contradicts the inclusion. Therefore, the systems remains at equilibrium in \emph{sliding mode}.
}\end{proof}

{\color{black}In the preceding \Cref{lemma_uniqueness}, notice that satisfying the basic conditions given in \cite{filippov1988differential} for every $(t,X)\in[0,\infty)\times\mathbb R^n_+$ allows to define for every $t\in[0,\infty)$ a solution of the differential inclusion
\begin{align*}
\frac{dX}{dt}\in\mathbf F(t,X).
\end{align*}}Finally, we have the elements to prove \Cref{proposition1}. 

\begin{proof}[Proof of \Cref{proposition1}]{
If there exists $t^*\geq t_0$ such that $X_i(t^*)=Y_i^*$ for ever{\color{black}y} $i=1,2,\dots,n$, then the {\color{black}switching} system \eqref{end_product_feedback} remains at the sliding mode, i.e.,  $X_i(t)=Y_i^*$ for all $t\geq t^*$ (see the proof of \Cref{lemma_uniqueness}). 

To continue with the proof, without loss of generality, we assume that $X(t)\not=Y^*$ for every $t\geq t_0$. Consider the ODE systems 
\begin{align}
\label{input_on_P1}
\frac{d\widetilde X}{dt}&=f(k_1,\widetilde X,\mu,\nu_n),\\
\label{input_on_K}
\frac{dY}{dt}&=f(\alpha\cdot k_1,Y,\mu,\nu_n),\\
\label{sys_inter}
\frac{d\mathcal X}{dt}&=f(\mathbf I,\mathcal X,\mu,\nu_n),\\
\label{input_off}
\frac{dZ}{dt}&=f(0,Z,\mu,\nu_n),
\end{align}where, in agreement with hypothesis \eqref{hyp_alpha}, 
$$\alpha:=\frac{1}{k_1}\big(\sum_{i=1}^{n-1}\mu\cdot Y_i^* + \mu\cdot K + \nu_nf_n(K)\big)<1$$and, by virtue of \Cref{lemma9}, $\mathbf I\in(\alpha\cdot k_1,k_1)$ is an input such that system \eqref{sys_inter} has an equilibrium point $\mathcal X^*:=(\mathcal X_1^*,\dots,\mathcal X_n^*)$ satisfying $Y_i^*<\mathcal X_i^*$ for all $i=1,2,\dots,n-1$ and 
\begin{align}\label{K_X_n}
K<\mathcal X_n^*.
\end{align}

Suppose that $X_n(t_0)=K$. Since $X(t_0)\not=Y^*$, as a consequence of \Cref{lemma11} and \Cref{lemma_uniqueness}, $X_n$ either increase{\color{black}s} or decrease{\color{black}s} after $t_0$. We will suppose that it decreases, i.e. $X_n(t)<K$ for $t\in(t_0,t_0+\varepsilon)$ and we will reset the initial condition in such way that $X_n(t_0)<K$ to continue with the demonstration. The reciprocal case when $X_n$ increases after $t_0$ can be proved analogously.

Hence, let $X_n(t_0)<K$ and consider system \eqref{input_on_P1} and \eqref{sys_inter} with the same initial conditions, i.e., $\widetilde X(t_0)=\mathcal X(t_0)=X(t_0)$. According to \Cref{lemma1} the equilibrium point $\mathcal X^*$ is GUAS. Then, \eqref{K_X_n} implies that there exists $t'>t_0$ such that $$K<\mathcal X_n(t').$$
On the other hand, since $\mathbf I<k_1$, according to \Cref{lemma5}, $\widetilde X$ {\color{black}upper} bounds $\mathcal X$. Then, $$K<\mathcal X_n(t')<\widetilde X_n(t')$$and we can assure that there exist{\color{black}s} $t_1\in(t_0,t{\color{black}'})$ such that 
\begin{align}\notag
\widetilde X_n(t)&<K&\forall t\in(t_0,t_1),\\\notag
\widetilde X_n(t_1)&=K,\\
\label{X_n_K}
\widetilde X_n(t)&>K&\forall t\in(t_1,t').
\end{align}
Hence, according to \Cref{lemma11} and \Cref{lemma_uniqueness},
\begin{align*}
X_n(t)&=\widetilde X_n(t)&\text{for a.e. }t\in[t_0,t_1].
\end{align*}Moreover, by continuity,
\begin{align*}
X_n(t_1)&=\widetilde X_n(t_1)=K,\\
\dot X_n(t_1)&=\dot{\widetilde X}_n(t_1)\geq0.
\end{align*}

Now consider system \eqref{input_off} with initial condition $Z(t_1)=X(t_1)$. Since we have supposed that $X(t)\not=Y^*$ for every $t\geq t_0$, it follows by \Cref{lemma11}, \Cref{lemma_uniqueness} and the inequality \eqref{X_n_K} that 
\begin{align*}
K&<Z_n(t)&\forall t\in(t_1,t_1+\varepsilon),\\
X_n(t)&=Z_n(t)&\text{for a.e. } t\in[t_1,t_1+\varepsilon],
\end{align*}for some $\varepsilon>0$. In other words, system \eqref{end_product_feedback} has switched at $t_1$ and follows the dynamics of system \eqref{input_off} in an interval $[t_1,t_1+\varepsilon]$.

Using similar arguments for the GUAS equilibrium point of system \eqref{input_off}, $\bar 0$, it can be proved that system \eqref{end_product_feedback} switches at a point $t_2>t_1$ and returns to the dynamics of \eqref{input_on_P1}. The oscillatory behavior of $X_i$ around $Y_i^*$ for $i=1,2,\dots,n-1$ can be proved by induction using as induction hypothesis that $X_j$ oscillates around $Y_j^*$ for every $j=i+1,\dots,n-1$ and $X_n$ oscillates around $K$.
}
\end{proof}

\subsection{Stable system}

In this Section, we present and prove \Cref{proposition2} that states some conditions under which the {\color{black}switching} system \eqref{end_product_feedback_st} converges uniformly and asymptotically to an equilibrium point. 

The {\color{black}switching} systems \eqref{end_product_feedback} and \eqref{end_product_feedback_st} are the same, but in \Cref{proposition2} we consider that the constant input ON is equal to or lower than the threshold value also considered in \Cref{proposition1}. Or that this threshold is not defined (because there is not a positive sliding mode) and that the system with the constant input ON has a nonnegative equilibrium point.

\begin{theorem}\label{proposition2}
Under Assumption \ref{assumption1}, consider the {\color{black}switching} system
\begin{align}\label{end_product_feedback_st}
\frac{dX}{dt}&=f(u[X_n],X,\mu,\nu_n)
\end{align}with initial conditions $X_i(t_0)\geq0$, $i=1,2,\dots,n$, $K>0$, $k_1>0$, $\mu\geq0$, $\nu_n\geq0$, $\mu+\nu_n>0$ and the input $u[X_n]$ defined in \eqref{input}. {\color{black}Consider the system with constant input ON \begin{align}\label{input_on}
\frac{d\widetilde X}{dt}&=f(k_1,\widetilde X,\mu,\nu_n).
\end{align}}

Suppose that one of the following conditions is satisfied:
\begin{itemize}
\item[i)] There exist positive values $Y_1^*,Y_2^*,\dots,Y_{n-1}^*$ such that 
\begin{align*}
0&=f_1(Y_1^*,Y_2^*)-f_2(Y_2^*,Y_3^*)-\mu\cdot Y_2^*\\\notag
0&=f_2(Y_2^*,Y_3^*)-f_3(Y_3^*,Y_4^*)-\mu\cdot Y_3^*\\\notag
\vdots\\\notag
0&=f_{n-1}(Y_{n-1}^*,K)-\nu_nf_n(K)-\mu\cdot K\notag,
\end{align*}and $$k_1\leq\sum_{i=1}^{n-1}\mu\cdot Y_i^* + \mu\cdot K + \nu_nf_n(K).$$
 \item[ii)] There are not positive values $Y_1^*,Y_2^*,\dots,Y_{n-1}^*$ such that 
\begin{align*}
0&=f_1(Y_1^*,Y_2^*)-f_2(Y_2^*,Y_3^*)-\mu\cdot Y_2^*\\\notag
0&=f_2(Y_2^*,Y_3^*)-f_3(Y_3^*,Y_4^*)-\mu\cdot Y_3^*\\\notag
\vdots\\\notag
0&=f_{n-1}(Y_{n-1}^*,K)-\nu_nf_n(K)-\mu\cdot K\notag,
\end{align*}and the system with constant input
{\color{black}\eqref{input_on}} has a nonnegative equilibrium point $\widetilde X^*:=(\widetilde X_1^*,\widetilde X_2^*,\dots,\widetilde X_n^*)$.\\
\end{itemize}

Then, there is an absolutely continuous solution $X=(X_1,X_2,\dots,X_n)$ that satisfies equation \eqref{end_product_feedback_st} for a.e. $t\in(t_0,\infty)$ and right uniqueness holds in $[t_0,\infty)\times[0,\infty)^n$. 

{\color{black}Furthermore, supposing that i) is satisfied,} if there is $t^*\geq t_0$ such that $X_i(t^*)=Y_i^*$ for every $i=1,2,\dots,n-1$ and $X_n(t^*)=K$, then the {\color{black}switching} system \eqref{end_product_feedback_st} remains at sliding mode, i.e., $X_i(t)=Y_i^*$ for every $i=1,2,\dots,n-1$ and $X_n(t)=K$ for all $t\geq t^*$. Otherwise, {\color{black}the system with constant input
{\color{black}\eqref{input_on}} has a nonnegative equilibrium point $\widetilde X^*:=(\widetilde X_1^*,\widetilde X_2^*,\dots,\widetilde X_n^*)$ and} $\widetilde X^*$ is globally u{\color{black}n}iformly asymptotically stable (GUAS) for the {\color{black}switching} system \eqref{end_product_feedback_st}. 

{\color{black}On the other hand, if ii) is satisfied, $\widetilde X^*$ is globally u{\color{black}n}iformly asymptotically stable (GUAS) for the {\color{black}switching} system \eqref{end_product_feedback_st}.}

\end{theorem}

{\color{black}
\begin{note}
Under the conditions of \Cref{proposition2}, the {\color{black}switching} system \eqref{end_product_feedback_st} can converge to the equilibrium point $Y^*:=(Y_1^*, Y_2^*,\dots, Y^*_{n-1}, K)$ corresponding to the sliding mode in two different ways. In one case, it converges in finite time to $Y^*$ when there is $t^*\geq t_0$ such that $X_i(t^*)=Y_i^*$ for all $i=1,2,\dots,n-1$ and $X_n(t^*)=K$ (we then say that it remains at sliding mode). In the other case, it converges asymptotically to $Y^*$ when $k_1=\sum_{i=1}^{n-1}\mu\cdot Y_i^* + \mu\cdot K + \nu_nf_n(K)$, because this implies that $Y^*=\widetilde X^*$, where $\widetilde X^*$ is the equilibrium point of the system with the constant input ON \eqref{input_on}.
\end{note}}

The proof of \Cref{proposition2} is given at the end of this section. The following \Cref{lemma10} allows to prove in \Cref{proposition2} that the system of the form \eqref{ode_meta_path} and the constant input ON has an equilibrium point lower or equal (entry by entry) to the equilibrium point corresponding to the sliding mode of the {\color{black}switching} system \eqref{end_product_feedback_st}. 

On the other hand, in \Cref{lemma8} it is shown that the {\color{black}switching} system is bounded by any system of the form \eqref{ode_meta_path} with constant input larger or equal to the constant input ON. {\color{black}The proofs of all Lemmas are in \Cref{supp_mat}.}

\begin{lemma}\label{lemma10}
Under Assumption \ref{assumption1}, suppose that there is a positive vector $X^*:=(X_1^*,X_2^*,\dots,X_n^*)$ ($0< X_i^*$ for all $i=1,\dots,n$) such that
\begin{align*}
0&=f_1(X^*_1,X^*_2)-f_2(X^*_2,X^*_3)-\mu\cdot X^*_2\\
0&=f_2(X^*_2,X^*_3)-f_2(X^*_3,X^*_4)-\mu\cdot X^*_3\\
\vdots\\\notag
0&=f_{n-1}(X^*_{n-1},X^*_n)-\nu_nf_n(X^*_n)-\mu\cdot X^*_n,
\end{align*}with $\mu\geq0$, $\nu_n\geq0$.

Then, for any $X_n'\in(0,X_n^*)$, there are unique $X_1',X_2',\dots,X_{n-1}'$ such that $0\leq X_i'<X_i^*$, for all $i=1,2,\dots,n$, and
\begin{align*}
0&=f_1(X'_1,X'_2)-f_2(X'_2,X'_3)-\mu\cdot X'_2\\
0&=f_2(X'_2,X'_3)-f_2(X'_3,X'_4)-\mu\cdot X'_3\\
\vdots\\\notag
0&=f_{n-1}(X'_{n-1},X'_n)-\nu_nf_n(X'_n)-\mu\cdot X'_n,
\end{align*}
\end{lemma}

\begin{lemma}\label{lemma8}
Under Assumption \ref{assumption1}, consider the {\color{black}switching} system
\begin{align*}
\frac{dX}{dt}&=f(u[X_n],X,\mu,\nu_n),
\end{align*}with initial conditions $X_i(t_0)\geq0$, $i=1,2,\dots,n$, $K>0$, $k_1>0$, $\mu\geq0$, $\nu_n\geq0$ and the input $u[X_n]$ defined in \eqref{input}. 

Then, the system 
\begin{align*}
\frac{d\widetilde X}{dt}&=f(\mathbf I,\widetilde X,\mu,\nu_n),
\end{align*}with constant input $\mathbf I\geq k_1$ and initial conditions $\widetilde X(t_0)=X(t_0)$, {\color{black}is an upper bound of the switching system, i.e.,}
\begin{align*}
X_i(t)&\leq\widetilde X_i(t)&\forall t_0\leq t, \forall i=1,2,\dots,n.
\end{align*}

Moreover, the system 
\begin{align*}
\frac{dZ}{dt}&=f(0,Z,\mu,\nu_n),
\end{align*}with constant input $0< k_1$ and initial conditions $Z(t_0)=X(t_0)$, {\color{black}is a lower bound of the switching} system, i.e.,
\begin{align*}
Z_i(t)&\leq X_i(t)&\forall t_0\leq t, \forall i=1,2,\dots,n.
\end{align*}
\end{lemma}

\begin{note}
\Cref{lemma8} can also be applied to the {\color{black}switching} system of \Cref{proposition1}.
\end{note}

There are now all the components necessaries to prove \Cref{proposition2}.

\begin{proof}[Proof of \Cref{proposition2}]{
The existence and uniqueness of the solution follow by \Cref{lemma_uniqueness}. To prove that $\widetilde X^*$ is GUAS for the {\color{black}switching} system \eqref{end_product_feedback_st}, we first consider the case where \emph{i)} is satisfied. Define 
$$\alpha:=\frac{1}{k_1}\big(\sum_{i=1}^{n-1}\mu\cdot Y_i^* + \mu\cdot K + \nu_nf_n(K)\big).$$ Notice that $1\leq\alpha$. First suppose that $1<\alpha$ and consider the system
\begin{align}\label{input_on_K_st}
\frac{dY}{dt}&=f(\alpha\cdot k_1,Y,\mu,\nu_n),
\end{align}which has the positive equilibrium point $Y^*:=(Y_1^*,Y_2^*,\dots,Y_{n-1}^*,K)$.  
By \Cref{lemma1}, $Y^*$ is globally uniformly asymptotically stable (GUAS) for \eqref{input_on_K_st}. Moreover, system \eqref{input_on} has also an equilibrium point $\widetilde X^*$ according to \Cref{lemma10}, because $k_1<\alpha\cdot k_1$. Moreover, by \Cref{lemma5},
\begin{align*}
\widetilde X_i(t)&<Y_i(t)&\forall t_0\leq t,\forall i=1,2,\dots,n,
\end{align*}$\widetilde X^*$ is a GUAS equilibrium point for system \eqref{input_on} by \Cref{lemma1} and $\widetilde X_n^*< K$ by \Cref{lemma2}. Then,there exists $T\geq t_0$ large enough such that 
\begin{align*}
\widetilde X_n(t)&<K&\forall t\geq T.
\end{align*}

On the other hand, according to \Cref{lemma8}, the {\color{black}switching} system \eqref{end_product_feedback_st} is bounded by the system with constant input \eqref{input_on}. That is to say,
\begin{align*}
X_i(t)&\leq\widetilde X_i(t)&\forall t_0\leq t, \forall i=1,2,\dots,n.
\end{align*}Therefore, 
\begin{align*}
X_i(t)&\leq\widetilde X_n(t)<K&\forall T\leq t.
\end{align*}

Then, the {\color{black}switching} system \eqref{end_product_feedback_st} is in the regime of the system with constant input \eqref{input_on} in the interval $[T,\infty)$. 
Therefore, $\widetilde X^*$ is a GUAS equilibrium point for the {\color{black}switching} system \eqref{end_product_feedback_st}.

For the case $\alpha=1$, suppose there is not $t'\geq t_0$ such that $X_n(t)\leq K$ for all $t\geq t'$. Then it can be proved that the {\color{black}switching} system \eqref{end_product_feedback_st} oscillates around the sliding mode $(Y_1^*,Y_2^*,\dots,Y_{n-1}^*,K)$ (in a similar way as in \Cref{proposition1}). This implies the existence of $t''\geq t_0$ such that $X_i(t'')\leq Y_i^*$ for every $i=1,2,\dots,n-1$ and $X_n(t'')\leq K$. But according to \Cref{lemma4}, system \eqref{input_on} is invariant in $[0,Y_1^*]\times[0,Y_2^*]\times\dots\times[0,Y_{n-1}^*]\times[0,K]$, because $\widetilde X_i^*=Y_i^*$ for $i=1,2,\dots,n-1$ and $\widetilde X^*=K$. Hence, $X_n(t)\leq K$ for every $t\geq t''$, which contradicts our supposition. 

We conclude that, when $\alpha=1$, there is $t'\geq t_0$ such that $X_n(t)\leq K$ for all $t\geq t'$. Then the switches system \eqref{end_product_feedback_st} follows the regime of system \eqref{input_on} and converges to the sliding mode $(Y_1^*,Y_2^*,\dots,Y_{n-1}^*,K)$, which is equal to the equilibrium point $\widetilde X^*$ of system \eqref{input_on}.

Now suppose that point \emph{i)} is not satisfied. Then, point \emph{ii)} holds by hypothesis. Moreover, $\widetilde X_n^*< K$ according to \Cref{lemma10}. Hence, there exists $T\geq t_0$ large enough such that $\widetilde X_n(t)<K$ for every $t\geq T$. But, by \Cref{lemma8}, the {\color{black}switching} system \eqref{end_product_feedback_st} is bounded by the system with constant input $k_1$ \eqref{input_on}, i.e.
\begin{align*}
X_i(t)&\leq\widetilde X_i(t) &\forall t_0\leq t.
\end{align*}In particular,
\begin{align*}
X_n(t)&\leq\widetilde X_n(t)<K &\forall T\leq t.
\end{align*}
Therefore, the {\color{black}switching} system \eqref{end_product_feedback_st} has the dynamics of system \eqref{input_on} in the interval $[T,\infty)$ and converges uniformly and asymptotically to $\widetilde X^*$.
}
\end{proof}

{\color{black}
\begin{note}
In this section we have presented results for a system where all the states have the same decay rate. A generalization of these results can be achieved considering different decay rates. This can be useful in the context of gene networks, where very distinct degradation rates are involved \cite{farcot2009periodic, farcot2010limit}.

We can then enunciate a more general result as follows:

\begin{theorem}\label{th_diff_decay}
Under Assumption \ref{assumption1}, consider the system
\begin{align}\label{diff_decay}
\frac{dX_1}{dt}&=u[X_n]-f_1(X_1,X_2)-\mu_1\cdot X_1\\\notag
\frac{dX_2}{dt}&=f_1(X_1,X_2)-f_2(X_2,X_3)-\mu_2\cdot X_2\\\notag
\vdots\\\notag
\frac{dX_n}{dt}&=f_{n-1}(X_{n-1},X_n)-\nu_nf_n(X_n)-\mu_n\cdot X_n\notag,
\end{align}with initial conditions $X_i(t_0)\geq0$, $i=1,2,\dots,n$, $K>0$, $k_1>0$, $\mu_i\geq0$ for all $i=1,2,\dots,n$, $\nu_n\geq0$, $\mu_n+\nu_n>0$ and the input $u[X_n]$ defined in \eqref{input}. 

Suppose that there exist positive values $Y_1^*,Y_2^*,\dots,Y_{n-1}^*$ such that 
\begin{align}\label{diff_decay_Y}
0&=f_1(Y_1^*,Y_2^*)-f_2(Y_2^*,Y_3^*)-\mu_2\cdot Y_2^*\\\notag
0&=f_2(Y_2^*,Y_3^*)-f_3(Y_3^*,Y_4^*)-\mu_3\cdot Y_3^*\\\notag
\vdots\\\notag
0&=f_{n-1}(Y_{n-1}^*,K)-\nu_nf_n(K)-\mu_n\cdot K\notag
\end{align}
and define 
\begin{align*}
\alpha:=\frac{\sum_{i=1}^{n-1}\mu_i\cdot Y_i^* + \mu_n\cdot K + \nu_nf_n(K)}{k_1}.
\end{align*}

Then, there is an absolutely continuous function $X=(X_1,X_2,\dots,X_n)$ that satisfies equation \eqref{diff_decay} for a.e. $t\in[t_0,\infty)$ and right uniqueness holds in $[t_0,\infty)\times[0,\infty)^n$. 

Furthermore, if there is $t^*\geq t_0$ such that $X_i(t^*)=Y_i^*$ for every $i=1,2,\dots,n-1$ and $X_n(t^*)=K$, then the {\color{black}switching} system \eqref{diff_decay} remains at sliding mode, i.e., for all $t\geq t^*$, $X_i(t)=Y_i^*$ for every $i=1,2,\dots,n-1$ and $X_n(t)=K$.

Otherwise, 
\begin{itemize}
\item if $\alpha<1$, the {\color{black}switching} system \eqref{diff_decay} oscillates around $(Y_1^*,Y_2^*,\dots,Y_{n-1}^*,K)$, 
\item if $1\leq\alpha$, the system with constant input
\begin{align}\label{diff_decay_k_1}
\frac{d\widetilde X_1}{dt}&=k_1-f_1(\widetilde X_1,\widetilde X_2)-\mu_1\cdot \widetilde X_1\\\notag
\frac{d\widetilde X_2}{dt}&=f_1(\widetilde X_1,\widetilde X_2)-f_2(\widetilde X_2,\widetilde X_3)-\mu_2\cdot \widetilde X_2\\\notag
\vdots\\\notag
\frac{d\widetilde X_n}{dt}&=f_{n-1}(\widetilde X_{n-1},\widetilde X_n)-\nu_nf_n(\widetilde X_n)-\mu_n\cdot \widetilde X_n\notag,
\end{align}has a nonnegative equilibrium point $\widetilde X^*$ and $\widetilde X^*$ is globally uniformly asymptotically stable (GUAS) for the {\color{black}switching} system \eqref{diff_decay}.\\
\end{itemize}

If such positive values $Y_1^*,Y_2^*,\dots,Y_{n-1}^*$ satisfying \eqref{diff_decay_Y} do not exist and system \eqref{diff_decay_k_1} has a nonnegative equilibrium point $\widetilde X^*$, then $\widetilde X^*$ is globally uniformly asymptotically stable (GUAS) for the {\color{black}switching} system \eqref{diff_decay}.

\end{theorem} 

The proof of \Cref{th_diff_decay} can be obtained following the proofs for \Cref{proposition1} and \Cref{proposition2}, and substituting $\mu\cdot X_i$, $\mu\cdot Y_i$, etc, for $\mu_i\cdot X_i$, $\mu_i\cdot Y_i$, etc, respectively. 

\end{note}
}

{\color{black}
\begin{note}
According to the definition given in \cite{wang2010conditions}, system \eqref{ode_meta_path} can be considered as a tridiagonal feedback system provided that $u$ and $f_i$ are of class $C^1$ for all $i=1,2,\dots,n$. Theorems 1 and 2 in \cite{wang2010conditions} give conditions to prove global asymptotic stability of an equilibrium point for tridiagonal feedback systems. These conditions imply the existence of a compact absorbing subset of the domain (i.e. an invariant compact set) and a Metzler and Hurwitz matrix that upper bounds the Jacobian matrix of the system or a second compound matrix. The examples of linear monotone tridiagonal systems with nonlinear negative feedback and the Goldbeter model are also used in \cite{wang2010conditions} to exhibit oscillations in tridiagonal monotone systems when the conditions for stability fail. 

Contrary to the differentiability condition required in \cite{wang2010conditions}, we consider a non continuous differential equation by assuming a negative feedback that is piecewise constant and discontinuous (see definition \eqref{input}). Our approach allows to obtain conditions to prove not only global uniform asymptotic stability of equilibrium points (\Cref{proposition2}), but also sliding mode and oscillatory regimes (\Cref{proposition1}). Moreover, these conditions are equivalent to solve an algebraic equation to find equilibria, which in case of mass-action kinetics is often feasible.
\end{note}
}

\section{Example with Michaelis-Menten reversible reactions}\label{section_rMM}

We show an example of the {\color{black}switching} system \eqref{end_product_feedback} (or system \eqref{end_product_feedback_st}) with 3 metabolites and Michaelis-Menten reversible reactions. Let
\begin{align*}
u[X_3]&:=\begin{cases}
k_1 & \text{ if } X_3< K \\
0 & \text{ if } K<X_3.
\end{cases}\\
f_1(X_1,X_2)&:=\frac{k_2X_1-l_2X_2}{m_2X_1+n_2X_2+K_2},\\
f_2(X_2,X_3)&:=\frac{k_3X_2-l_3X_3}{m_3X_2+n_3X_3+K_3},\\
f_3(X_3)&:=\nu_n\frac{X_3}{X_3+K_n},
\end{align*}
\begin{align*}
f(u,X,\mu,\nu_n):=&\begin{pmatrix}
u - f_1(X_1,X_2) - \mu X_1\\
f_1(X_1,X_2) - f_2(X_2,X_3) - \mu X_2\\
f_2(X_2,X_3) - \nu_nf_3(X_3) - \mu X_3
\end{pmatrix},
\end{align*}for some $k_1>0$, $K>0$, $\mu\geq0$, $\nu_n\geq0$ and $0<\mu+\nu_n$.

Consider the ODE systems
\begin{align}\label{rMM}
\frac{dX}{dt}&=f(u[X_3],X,\mu,\nu_n),\\
\label{rMM2k1}
\frac{d\widetilde X}{dt}&=f(k_1,\widetilde X,\mu,\nu_n),\\
\label{rMMalpha}
\frac{dY}{dt}&=f(\alpha\cdot k_1,Y,\mu,\nu_n),\\
\label{rMMzero}
\frac{dZ}{dt}&=f(0,Z,\mu,\nu_n),
\end{align}where 
$$\alpha:=\frac{1}{k_1}\big(\sum_{i=1}^{n-1}\mu\cdot Y_i^* + \mu\cdot K + \nu_nf_n(K)\big).$$

\subsection{Oscillatory system}

\Cref{fig:osc1} and \Cref{fig:osc2} show two examples of the oscillatory behavior described in \Cref{proposition1}. In both cases, the constant input ON of the {\color{black}switching} system \eqref{rMM} is larger enough to satisfy the inequality in \eqref{hyp_alpha}. Moreover, we observe that the solution of the {\color{black}switching} system \eqref{rMM} is bounded between the solution of system \eqref{rMM2k1} with the constant input ON and the solution of system \eqref{rMMzero} with the constant input OFF, as stated in \Cref{lemma8}. 

Notice that even when system \eqref{rMM2k1} with the constant input ON has no positive equilibrium point, the {\color{black}switching} system \eqref{rMM} oscillates around the sliding mode (see \Cref{fig:osc2}). 

{\color{black}On the other hand, in \Cref{app_pasternack} we present an example to compare the dynamics of a piecewise linear model and a hybrid model as \eqref{end_product_feedback}. The example shows oscillations in both cases for a pathway with irreversible kinetics.}
 
\begin{figure}[!]
\includegraphics[width=\columnwidth]{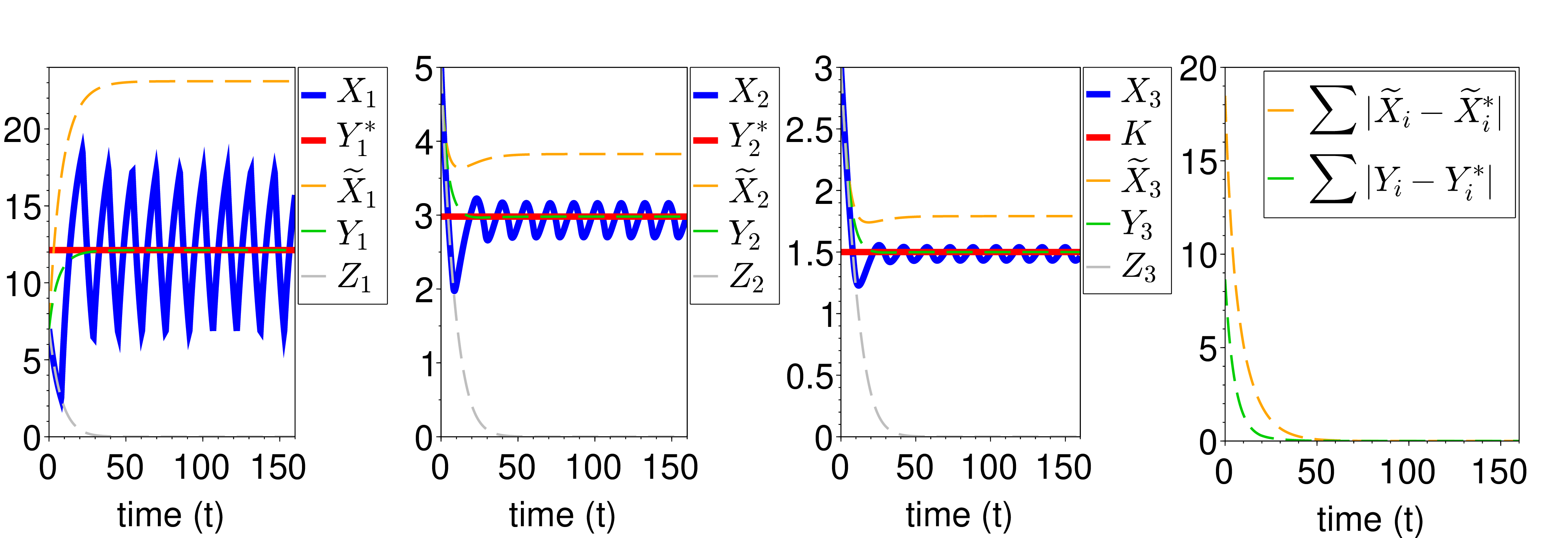}
\caption{Solution for systems \eqref{rMM}-\eqref{rMMzero} and candidate Lyapunov function for \eqref{rMM2k1} and \eqref{rMMalpha}. All parameters are equal to $1$ except for $k_1=3$, $\nu_n=0.2$, $\mu=0.1$ and $K=1.5$. The initial conditions are $X_1(0)=7$, $X_2(0)=5$ and $X_3(0)=3$. 
The {\color{black}switching} system \eqref{rMM} oscillates around the sliding mode $(Y_1^*,Y_2^*,K)$, which is the equilibrium point of system \eqref{rMMalpha}. Moreover, the {\color{black}switching} system \eqref{rMM} is bounded between systems \eqref{rMM2k1} and \eqref{rMMzero}, which converge both to their respective equilibrium points.}\label{fig:osc1}
\end{figure} 

\begin{figure}[!]
\includegraphics[width=\columnwidth]{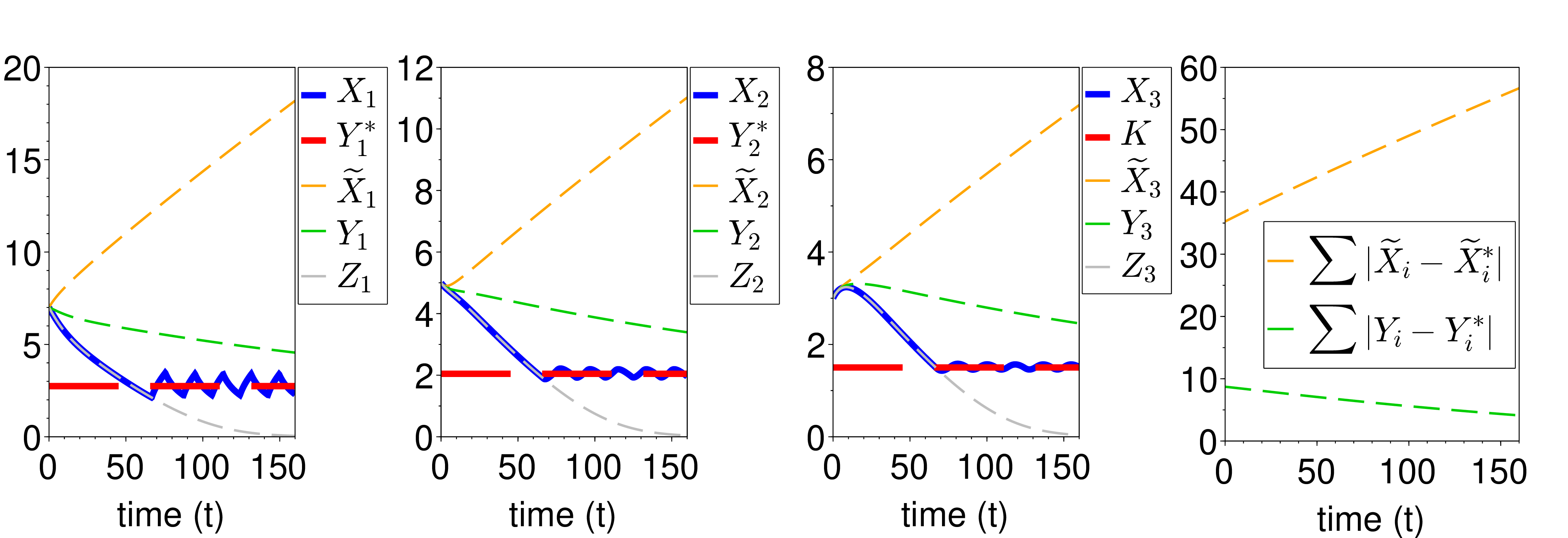}
\caption{Solution for systems \eqref{rMM}-\eqref{rMMzero} and candidate Lyapunov function for \eqref{rMM2k1} and \eqref{rMMalpha}. All parameters are equal to $1$ except for $k_1=0.3$, $\nu_n=0.2$, $\mu=0$ and $K=1.5$. The initial conditions are $X_1(0)=7$, $X_2(0)=5$ and $X_3(0)=3$. The {\color{black}switching} system \eqref{rMM} oscillates around the sliding mode $(Y_1^*,Y_2^*,K)$, which is the equilibrium point of system \eqref{rMMalpha}. Moreover, the {\color{black}switching} system \eqref{rMM} is bounded between systems \eqref{rMM2k1} and \eqref{rMMzero}. However, system \eqref{rMM2k1} {\color{black}does not have a positive equilibrium point and it is not stable.}. System \eqref{rMMzero} converges to $\bar0$.}\label{fig:osc2}
\end{figure} 

\subsection{Stable system}

In \Cref{fig:st1}, \Cref{fig:st2} and \Cref{fig:st3} there are three examples of the stabilization of the {\color{black}switching} system \eqref{rMM} as stated in \Cref{proposition2}.

\Cref{fig:st1} depicts the case when system \eqref{rMMalpha} with constant input $\alpha\cdot k_1$ has a positive equilibrium point and $1<\alpha$ (i.e. \emph{i)} is satisfied). In this case, the constant input ON is small enough to let the system stabilize and do not oscillate around the sliding mode. Notice that the equilibrium point of the {\color{black}switching} system \eqref{rMM} is uniformly asymptotically stable and lower (entry by entry) than the equilibrium point related to the sliding mode (i.e. $Y^*$ the equilibrium point of system \eqref{rMMalpha}). 

In \Cref{fig:st2}, condition \emph{i)} is not satisfied, but \emph{ii)} holds. That is to say, the system related to the sliding mode \eqref{rMMalpha} has not a nonnegative equilibrium point (it has an equilibrium point, but its first entry is negative) and system \eqref{rMM2k1} with the constant input ON has a positive equilibrium point. In this case, the {\color{black}switching} system \eqref{rMM} converges uniformly and asymptotically to the equilibrium point of system \eqref{rMM2k1}.

The case when the {\color{black}switching} system \eqref{rMM} reaches the sliding mode is represented in \Cref{fig:st3}. Here the constant input ON is such that system \eqref{rMM2k1} has an equilibrium point satisfying $\widetilde X_n^*=K$. The {\color{black}switching} system \eqref{rMM} converges uniformly and asymptotically to the sliding mode $Y^*=\widetilde X^*$.

Finally, as stated in \Cref{lemma8}, in the three examples it can be observed that the solution of the {\color{black}switching} system \eqref{rMM} is upper and lower bounded by system \eqref{rMM2k1} and system \eqref{rMMzero}, the solutions of the systems with the constant inputs ON and OFF, respectively.

\begin{figure}[!]
\includegraphics[width=\columnwidth]{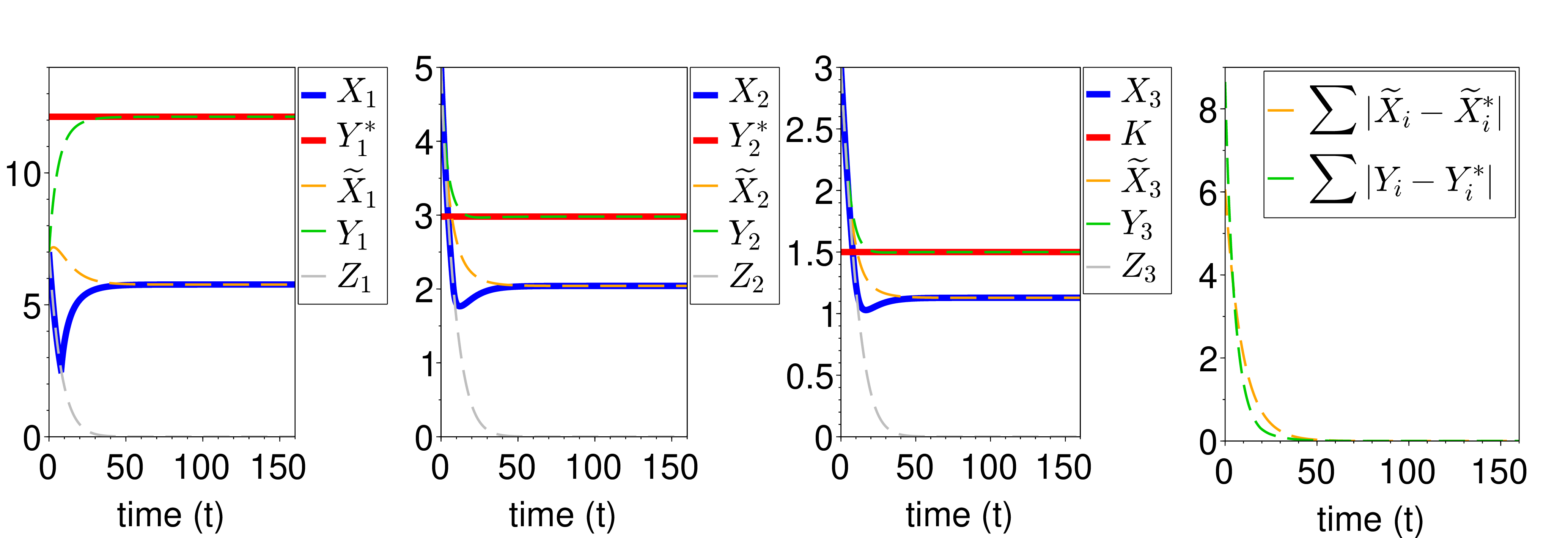}
\caption{Solution for systems \eqref{rMM}-\eqref{rMMzero} and candidate Lyapunov function for \eqref{rMM2k1} and \eqref{rMMalpha}. All parameters are equal to $1$ except for $\nu_n=0.2$, $\mu=0.1$ and $K=1.5$. The initial conditions are $X_1(0)=7$, $X_2(0)=5$ and $X_3(0)=3$. The {\color{black}switching} system \eqref{rMM} converges to the equilibrium point of system \eqref{rMM2k1}, because $k_1<\alpha\cdot k_1$. Moreover, the {\color{black}switching} system \eqref{rMM} is bounded between systems \eqref{rMM2k1} and \eqref{rMMzero}, which converge both to their respective equilibrium points. System \eqref{rMMalpha} converges to its positive equilibrium point $(Y_1^*,Y_2^*,K)$.}\label{fig:st1}
\end{figure} 

\begin{figure}[!]
\includegraphics[width=\columnwidth]{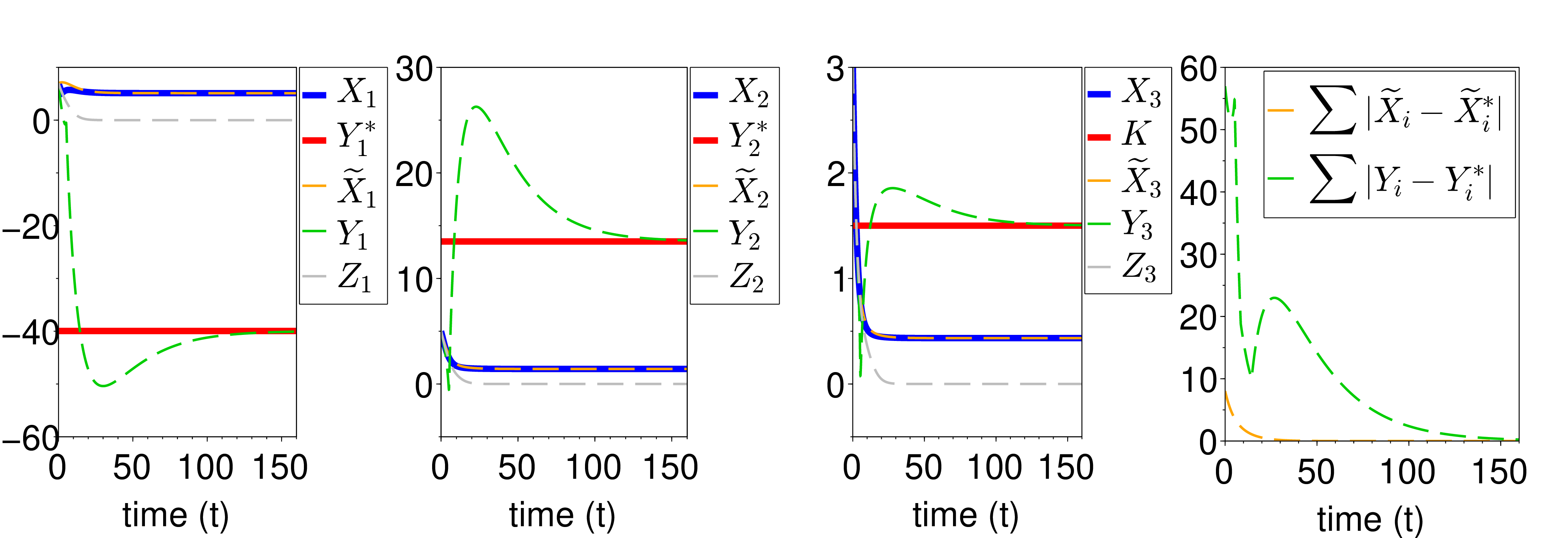}
\caption{Solution for systems \eqref{rMM}-\eqref{rMMzero} and candidate Lyapunov function for \eqref{rMM2k1} and \eqref{rMMalpha}. All parameters are equal to $1$ except for $\mu=0.1$ and $K=1.5$. The initial conditions are $X_1(0)=7$, $X_2(0)=5$ and $X_3(0)=3$. The {\color{black}switching} system \eqref{rMM} converges to the equilibrium point of system \eqref{rMM2k1}, because $\alpha<1$. Moreover, the {\color{black}switching} system \eqref{rMM} is bounded between systems \eqref{rMM2k1} and \eqref{rMMzero}, which converge both to their respective equilibrium points. The equilibrium point $(Y_1^*,Y_2^*,K)$ of system \eqref{rMMalpha} is nonpositive.}\label{fig:st2}
\end{figure}

\begin{figure}[!]
\includegraphics[width=\columnwidth]{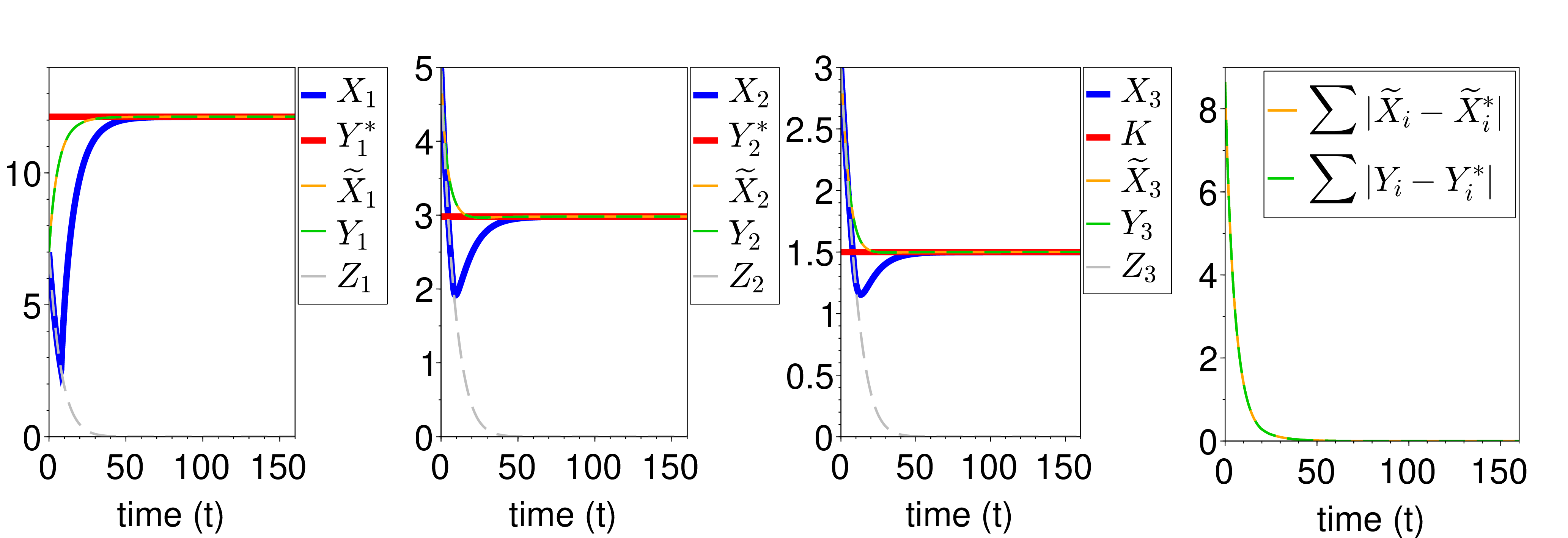}
\caption{Solution for systems \eqref{rMM}-\eqref{rMMzero} and candidate Lyapunov function for \eqref{rMM2k1} and \eqref{rMMalpha}. All parameters are equal to $1$ except for $k_1=\mu\cdot (Y_1^* + Y_2^* + K) + \nu_nf_n(K)\sim1.78$, $\nu_n=0.2$, $\mu=0.1$ and $K=1.5$. The initial conditions are $X_1(0)=7$, $X_2(0)=5$ and $X_3(0)=3$. The {\color{black}switching} system \eqref{rMM} converges {\color{black}asymptotically} to the sliding mode, which is the equilibrium point of systems \eqref{rMM2k1} and \eqref{rMMalpha}, because $\alpha=1$. Moreover, the {\color{black}switching} system \eqref{rMM} is bounded between systems \eqref{rMM2k1} and \eqref{rMMzero}, which converge both to their respective equilibrium points.}\label{fig:st3}
\end{figure} 
\newpage
{\color{black}
\subsection{Example with different decay rates}
We show an example of a system with different decay rates to illustrate \Cref{th_diff_decay}. As in the previous example, we consider reversible Michaelis-Menten kinetics. Let
\begin{align*}
u[X_3]&:=\begin{cases}
k_1 & \text{ if } X_3< K \\
0 & \text{ if } K<X_3.
\end{cases}\\
f_1(X_1,X_2)&:=\frac{k_2X_1-l_2X_2}{m_2X_1+n_2X_2+K_2},\\
f_2(X_2,X_3)&:=\frac{k_3X_2-l_3X_3}{m_3X_2+n_3X_3+K_3},\\
f_3(X_3)&:=\nu_n\frac{X_3}{X_3+K_n},
\end{align*}
\begin{align*}
g(u,X,\overline\mu,\nu_n):=&\begin{pmatrix}
u - f_1(X_1,X_2) - \mu_1 X_1\\
f_1(X_1,X_2) - f_2(X_2,X_3) - \mu_2 X_2\\
f_2(X_2,X_3) - \nu_nf_3(X_3) - \mu_3 X_3
\end{pmatrix},
\end{align*}for some $k_1>0$, $K>0$, $\mu_i\geq0$ for all $i=1,2,3$, $\nu_n\geq0$ and $0<\mu_3+\nu_n$.

Consider the ODE systems
\begin{align}\label{rMM_decay}
\frac{dX}{dt}&=g(u[X_3],X,\overline\mu,\nu_n),\\
\label{rMM2k1_decay}
\frac{d\widetilde X}{dt}&=g(k_1,\widetilde X,\overline\mu,\nu_n),\\
\label{rMMalpha_decay}
\frac{dY}{dt}&=g(\alpha\cdot k_1,Y,\overline\mu,\nu_n),\\
\label{rMMzero_decay}
\frac{dZ}{dt}&=g(0,Z,\overline\mu,\nu_n),
\end{align}where 
$$\alpha:=\frac{1}{k_1}\big(\sum_{i=1}^{n-1}\mu_i\cdot Y_i^* + \mu_n\cdot K + \nu_nf_n(K)\big).$$

\begin{figure}[!]
\includegraphics[width=\columnwidth]{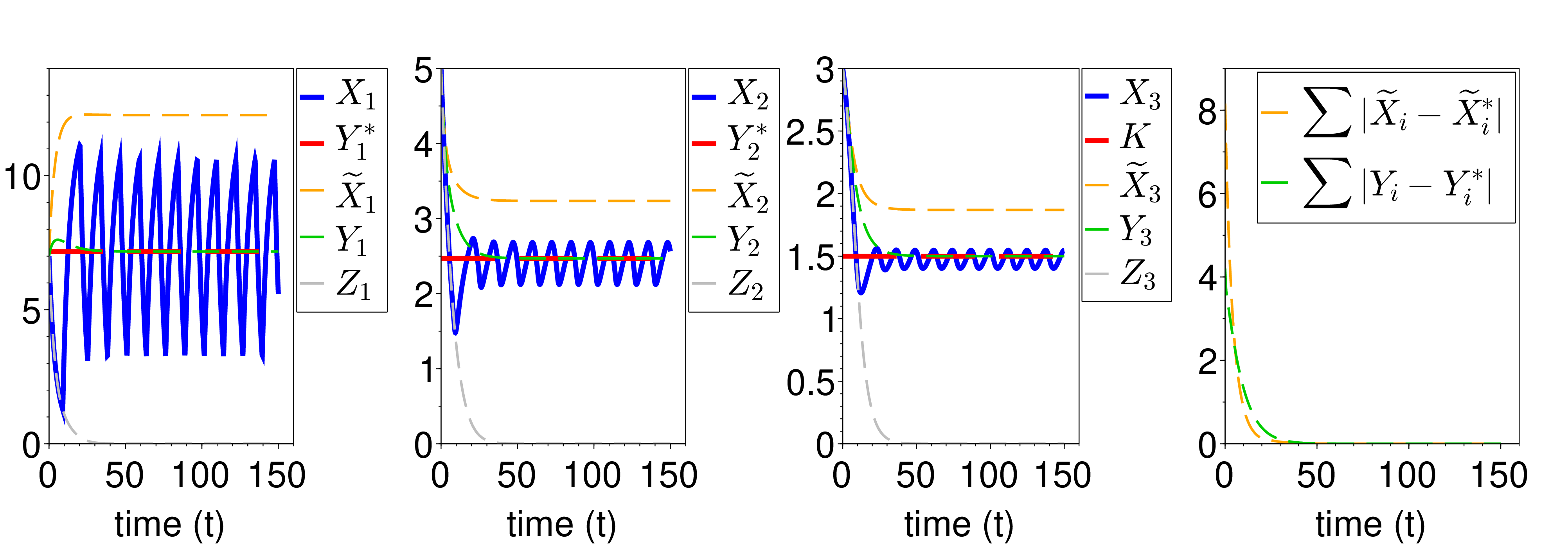}
\caption{\color{black}Solution for systems \eqref{rMM_decay}-\eqref{rMMzero_decay} with different decay rates and candidate Lyapunov function for \eqref{rMM2k1_decay} and \eqref{rMMalpha_decay}. All parameters are equal to $1$ except for $k_1=3$, $\nu_n=0.2$, $\mu_1=0.2$, $\mu_2=0.1$, $\mu_3=0.05$ and $K=1.5$. The initial conditions are $X_1(0)=7$, $X_2(0)=5$ and $X_3(0)=3$. 
The {\color{black}switching} system \eqref{rMM_decay} oscillates around the sliding mode $(Y_1^*,Y_2^*,K)$, which is the equilibrium point of system \eqref{rMMalpha_decay}. Moreover, the {\color{black}switching} system \eqref{rMM_decay} is bounded between systems \eqref{rMM2k1_decay} and \eqref{rMMzero_decay}, which converge both to their respective equilibrium points.}
\end{figure} 

\begin{figure}[!]
\includegraphics[width=\columnwidth]{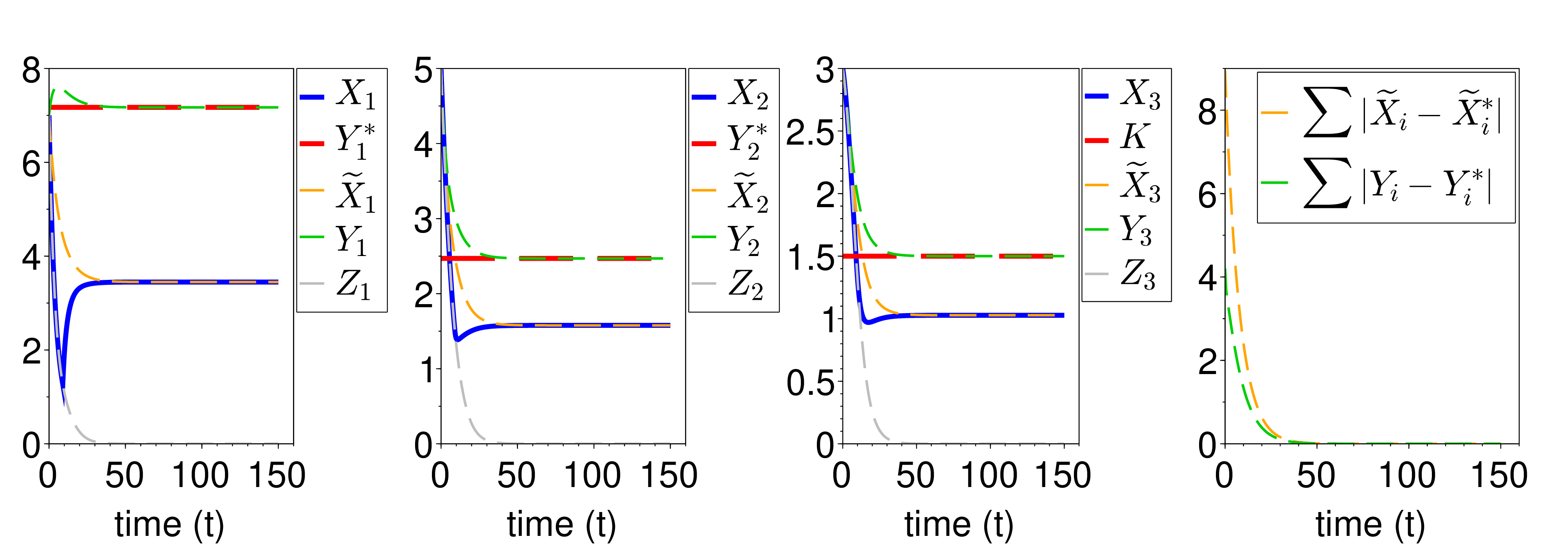}
\caption{\color{black}Solution for systems \eqref{rMM_decay}-\eqref{rMMzero_decay} with different decay rates and candidate Lyapunov function for \eqref{rMM2k1_decay} and \eqref{rMMalpha_decay}. All parameters are equal to $1$ except for $\nu_n=0.2$, $\mu_1=0.2$, $\mu_2=0.1$, $\mu_3=0.05$ and $K=1.5$. The initial conditions are $X_1(0)=7$, $X_2(0)=5$ and $X_3(0)=3$. The initial conditions are $X_1(0)=7$, $X_2(0)=5$ and $X_3(0)=3$. The {\color{black}switching} system \eqref{rMM_decay} converges to the equilibrium point of system \eqref{rMM2k1_decay}, because $k_1<\alpha\cdot k_1$. Moreover, the {\color{black}switching} system \eqref{rMM_decay} is bounded between systems \eqref{rMM2k1_decay} and \eqref{rMMzero_decay}, which converge both to their respective equilibrium points. System \eqref{rMMalpha_decay} converges to its positive equilibrium point $(Y_1^*,Y_2^*,K)$.}
\end{figure} 

\begin{figure}[!]
\includegraphics[width=\columnwidth]{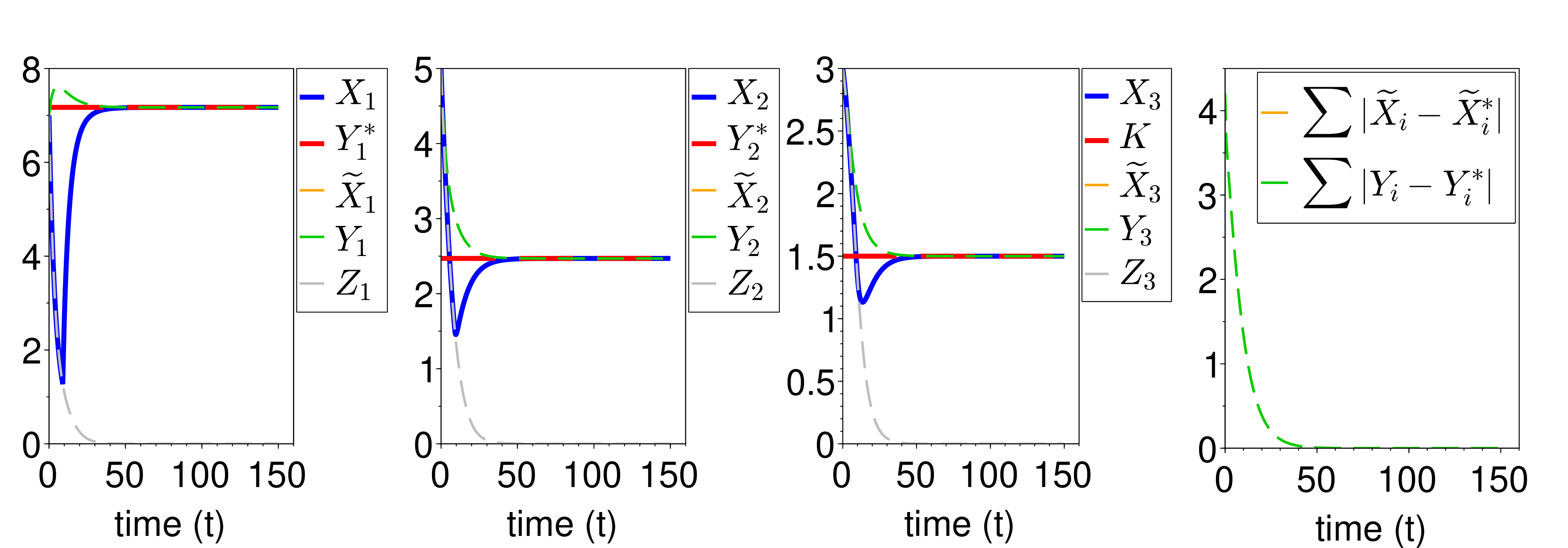}
\caption{\color{black}Solution for systems \eqref{rMM_decay}-\eqref{rMMzero_decay} with different decay rates and candidate Lyapunov function for \eqref{rMM2k1_decay} and \eqref{rMMalpha_decay}. All parameters are equal to $1$ except for $k_1=\mu_1\cdot Y_1^* + \mu_2\cdot Y_2^* + \mu_3\cdot K + \nu_nf_n(K)\sim1.88$, $\nu_n=0.2$, $\mu_1=0.2$, $\mu_2=0.1$, $\mu_3=0.05$ and $K=1.5$. The initial conditions are $X_1(0)=7$, $X_2(0)=5$ and $X_3(0)=3$. The {\color{black}switching} system \eqref{rMM_decay} converges asymptotically to the sliding mode, which is the equilibrium point of systems \eqref{rMM2k1_decay} and \eqref{rMMalpha_decay}, because $\alpha=1$. Moreover, the {\color{black}switching} system \eqref{rMM_decay} is bounded between systems \eqref{rMM2k1_decay} and \eqref{rMMzero_decay}, which converge both to their respective equilibrium points.}
\end{figure} 
 }

\section{Continuous feedback systems}\label{sec_continuous}

The purpose of this Section is to exhibit some sequences of equations of the form \eqref{ode_meta_path} with continuous inputs whose solutions converge pointwise to the solution of the {\color{black}switching} system \eqref{end_product_feedback} (equal to system \eqref{end_product_feedback_st}). This follows the idea of considering the switching input \eqref{input} as the limit case of allosteric regulation processes that occur very fast (see \Cref{introduction} and \Cref{repress}). 

\subsection{Smooth input}

In \Cref{smooth}, we introduce a sequence of equations of the form \eqref{ode_meta_path} with smooth inputs. The purpose is that the smooth inputs converge to the step function defined by the switching input \eqref{input}. For this, sigmoid functions of the form 
$$\frac{k_1}{1+\big(\frac{X_n}{K}\big)^m}$$are considered. However, a sigmoid input would not allow the system to {\color{black}remain constant in the sliding mode if there is $t^*\geq t_0$ such that $X_i(t^*)=Y_i^*$ for every $i=1,2,\dots,n-1$ and $X_n(t^*)=K$. This is a property that, according to \Cref{proposition1} and \Cref{proposition2}, the {\color{black}switching} systems \eqref{end_product_feedback} and \eqref{end_product_feedback_st} satisfy. In order to approximate this particular dynamics of the switching systems}, the sigmoid function is multiplied by the Gaussian function
\begin{align*}
\Big[(2\cdot\alpha-1)\cdot\exp\Big\{-\Big(\frac{X_n-K}{\sigma}\Big)^2\Big\}+1\Big].
\end{align*}

\begin{proposition}\label{smooth}
Under Assumption \ref{assumption1}, consider the {\color{black}switching} system
\begin{align*}
\frac{dX}{dt}&=f(u[X_n],X,\mu,\nu_n)
\end{align*}with initial conditions $X_i(t_0)\geq0$, $i=1,2,\dots,n$, $K>0$, $k_1>0$, $\mu\geq0$, $\nu_n\geq0$ and the input $u[X_n]$ defined in \eqref{input}. 

Suppose that there are positive values $Y_1^*,Y_2^*,\dots,Y_{n-1}^*$ such that 
\begin{align*}
0&=f_1(Y_1^*,Y_2^*)-f_2(Y_2^*,Y_3^*)-\mu\cdot Y_2^*\\\notag
0&=f_2(Y_2^*,Y_3^*)-f_3(Y_3^*,Y_4^*)-\mu\cdot Y_3^*\\\notag
\vdots\\\notag
0&=f_{n-1}(Y_{n-1}^*,K)-\nu_nf_n(K)-\mu\cdot K\notag,
\end{align*}and define
\begin{align*}
 \alpha:=\frac{1}{k_1}\Big(\sum_{i=1}^{n-1}\mu\cdot Y_i^* + \mu\cdot K + \nu_nf_n(K)\Big).
\end{align*}

For every $m\in\mathbb N$, let $\varphi^m:=(\varphi^m_1,\varphi^m_2,\dots,\varphi^m_n)$ be the solution for the continuous differential equation
\begin{align*}
\frac{d\varphi^m}{dt}&=f\big(U_m^\sigma(\varphi^m_n), \varphi^m,\mu,\nu\big)
\end{align*}where 
\begin{align*}
U_m^\sigma(\varphi^m_n):=\frac{k_1}{1+\big(\frac{\varphi^m_n}{K}\big)^m}\cdot \Big[(2\cdot\alpha-1)\cdot\exp\Big\{-\Big(\frac{\varphi^m_n-K}{\sigma}\Big)^2\Big\}+1\Big],
\end{align*}$\sigma$ is a small real number and the initial conditions $\varphi^m_i(t_0)=X_i(t_0)$ for all $i=1,2,\dots,n$.

{\color{black}Then}, for a.e. $t\in[t_0,\infty)$,
\begin{align*}
\lim_{m\to\infty,\sigma\to0}\varphi^m_i(t,\sigma)&=X_i(t)&\forall i=1,2,\dots,n.
\end{align*}
\end{proposition}

\begin{proof}
Notice that
\begin{align*}
\lim_{m\to\infty,\sigma\to0}
U_m^\sigma(\varphi^m_n)=\begin{cases}
k_1&\text{if } \varphi^m_n<K\\
0&\text{if }\varphi^m_n>K\\
\alpha\cdot k_1&\text{if }\varphi^m_n=K
\end{cases}.
\end{align*}
The result follows from Lemma 3 in Section 7, Chapter 2 (p. 82) of \cite{filippov1988differential}.
\end{proof}

\begin{example}
Consider the {\color{black}switching} system \eqref{rMM} with Michaelis-Menten kinetics introduced in \Cref{section_rMM}. {\color{black} In \Cref{smooth_oscillation}, \Cref{smooth_stable} and \Cref{smooth_sliding} is depicted the approximation of the {\color{black}switching} system \eqref{rMM} given by the series of function $\{\varphi^m\}$. Note that to closely approximate the solution of \eqref{rMM}, the parameter $m$ in the series of smooth equations has to be large, specially in case of sliding mode (see \Cref{smooth_sliding}).}

\begin{figure}[!]
\includegraphics[width=\columnwidth]{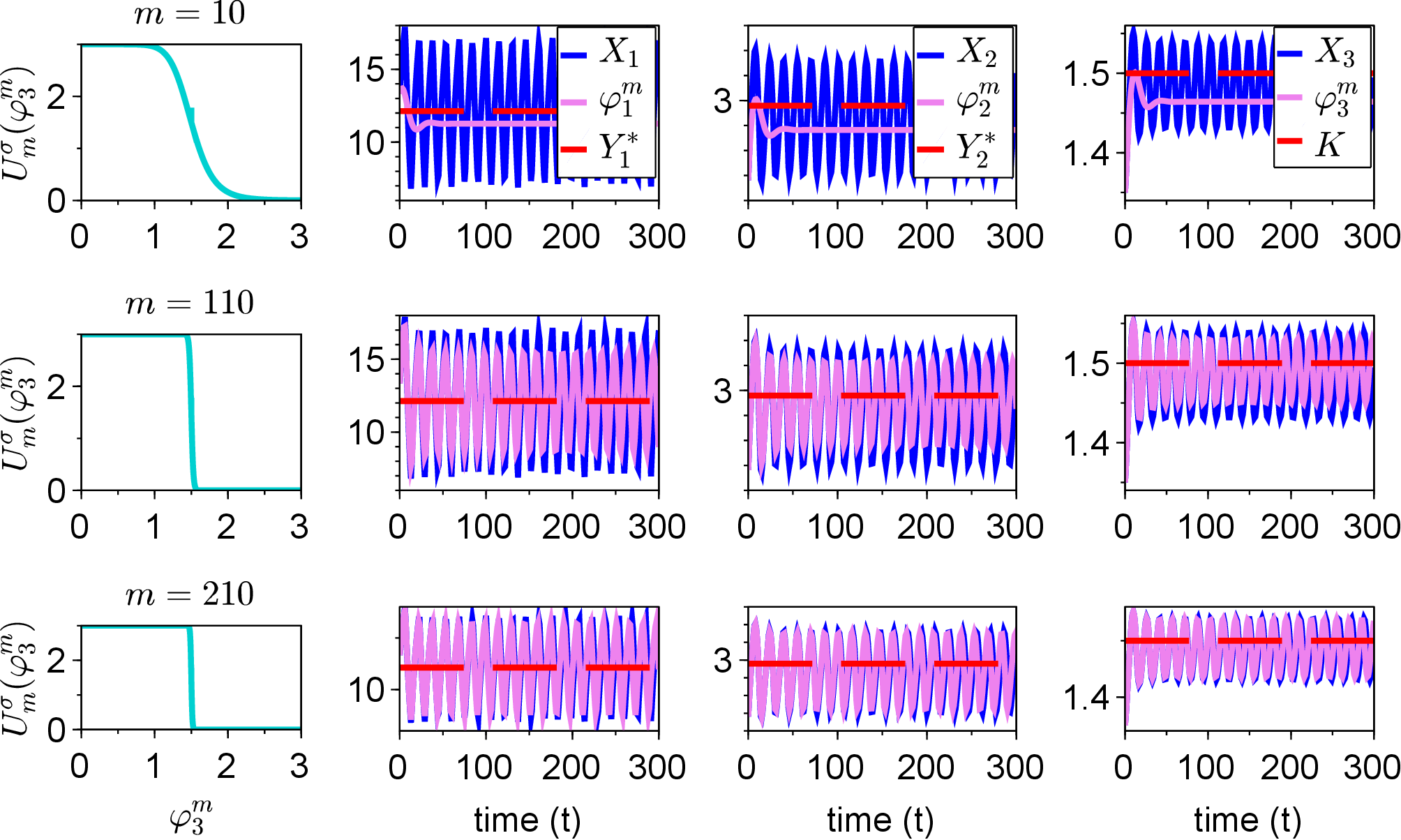}
\caption{The {\color{black}switching} system \eqref{rMM} oscillates around the sliding mode $(Y_1^*, Y_2^*,K)$. Functions $\varphi^m$ approximate $X$ as $m\to\infty$ and $\sigma\to0$. All parameters are equal to $1$ except for $k_1=3$, $\nu_n=0.2$, $\mu=0.1$, $K=1.5$ and $\sigma=10^{-6}$. The initial conditions are $X_1(0)=13.34$, $X_2(0)=2.68$ and $X_3(0)=1.35$. }\label{smooth_oscillation}
\end{figure}

\begin{figure}[!]
\includegraphics[width=\columnwidth]{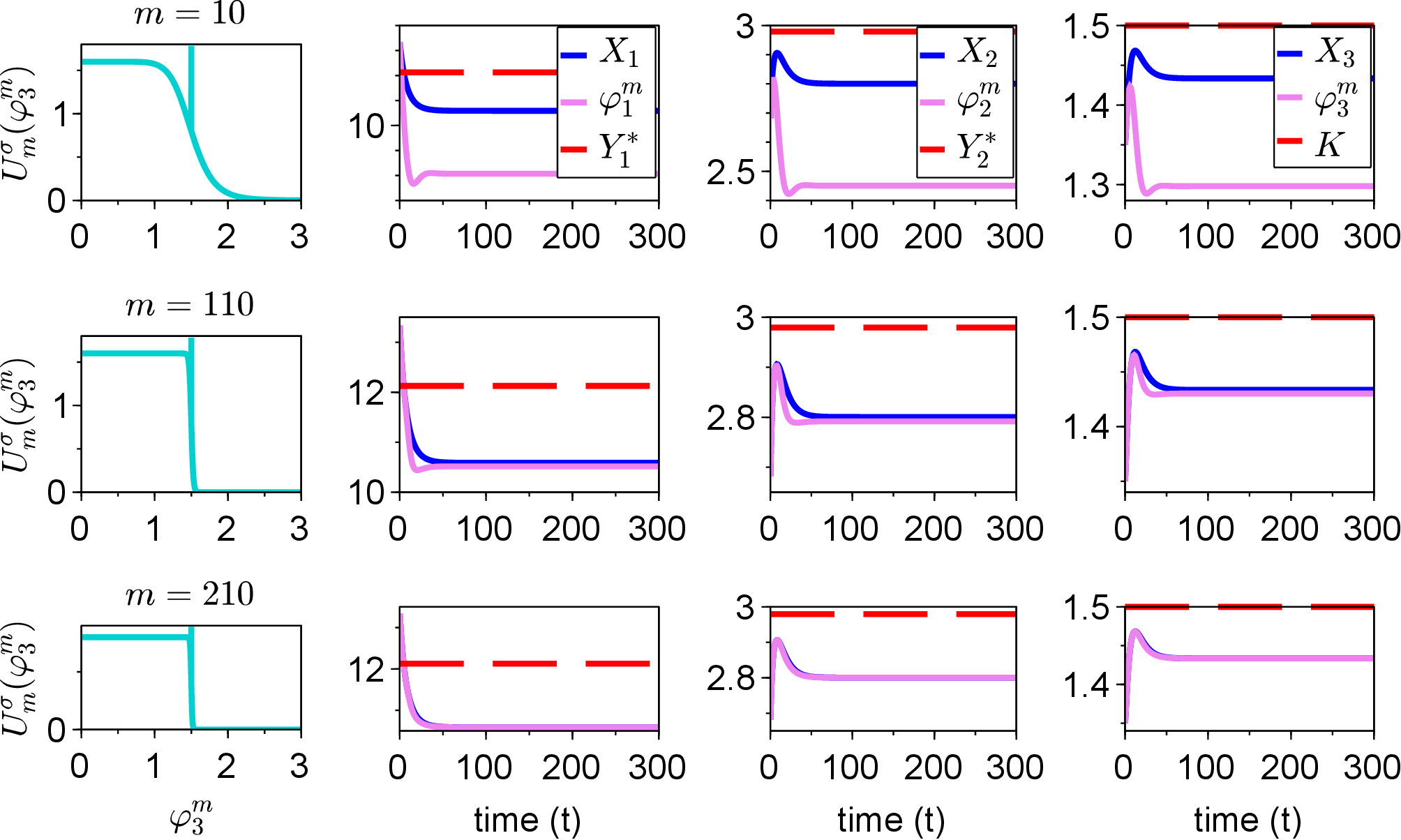}
\caption{The {\color{black}switching} system \eqref{rMM} converges to an equilibrium point {\color{black}smaller than} the sliding mode $(Y_1^*, Y_2^*,K)$. Functions $\varphi^m$ approximate $X$ as $m\to\infty$ and $\sigma\to0$. All parameters are equal to $1$ except for $k_1=1.6$, $\nu_n=0.2$, $\mu=0.1$, $K=1.5$ and $\sigma=10^{-6}$. The initial conditions are $X_1(0)=13.34$, $X_2(0)=2.68$ and $X_3(0)=1.35$. }\label{smooth_stable}
\end{figure}

\begin{figure}[!]
\includegraphics[width=\columnwidth]{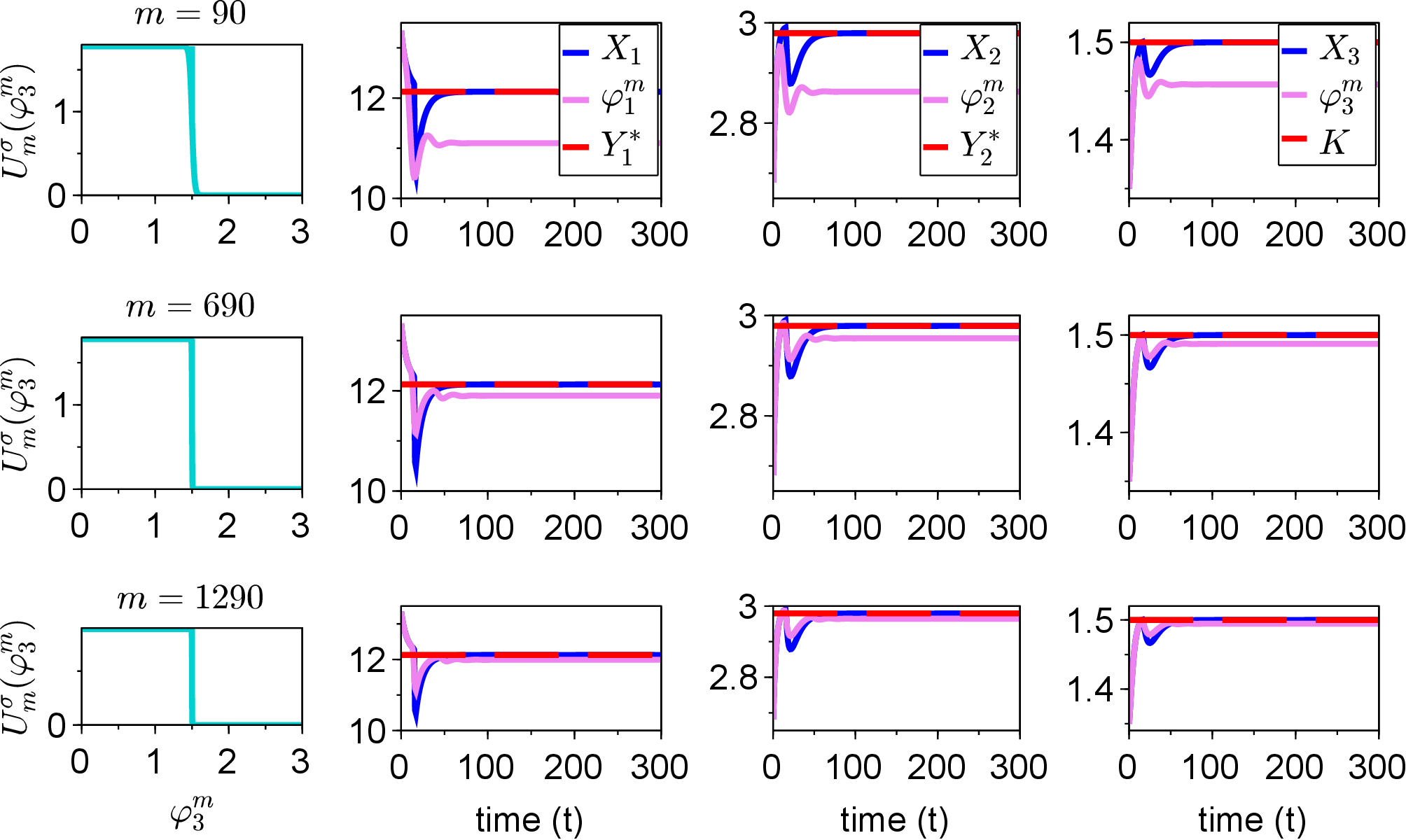}
\caption{The {\color{black}switching} system \eqref{rMM} converges to the sliding mode $(Y_1^*, Y_2^*,K)$. Functions $\varphi^m$ approximate $X$ as $m\to\infty$ and $\sigma\to0$. All parameters are equal to $1$ except for $k_1=\sum_{i=1}^{n-1}\mu\cdot Y_i^* + \mu\cdot K + \nu_nf_n(K)\sim 1.78$, $\nu_n=0.2$, $\mu=0.1$, $K=1.5$ and $\sigma=10^{-6}$. The initial conditions are $X_1(0)=13.34$, $X_2(0)=2.68$ and $X_3(0)=1.35$. }\label{smooth_sliding}
\end{figure} 

\end{example}

\subsection{Piecewise linear input}

In \Cref{continuous}, we introduce a sequence of equations with continuous inputs (even so not everywhere differentiable). The idea is to approximate the step function defined by the switching input \eqref{input} by a sequence of piecewise linear functions of the form
\begin{align*}
L_\varepsilon\big[X_n\big]:=\begin{cases}
k_1&\text{ if }X_n<K-\varepsilon\\
\frac{k_1}{2\cdot\varepsilon}(K+\varepsilon-X_n)&\text{ if }K-\varepsilon\leq X_n\leq K+\varepsilon\\
0&\text{ if }X_n>K+\varepsilon\\
\end{cases}
\end{align*}with $\varepsilon$ a small positive number. However, the input $L_\varepsilon$ above would not allow the functions of the sequence to {\color{black}remain constant in the sliding mode if there is $t^*\geq t_0$ such that $X_i(t^*)=Y_i^*$ for every $i=1,2,\dots,n-1$ and $X_n(t^*)=K$. This is a property that, according to \Cref{proposition1} and \Cref{proposition2}, the {\color{black}switching} systems \eqref{end_product_feedback} and \eqref{end_product_feedback_st} satisfy. In order to approximate this particular dynamics of the switching systems}, a slightly more elaborate piecewise linear function is defined \Cref{continuous}.  

\begin{proposition}\label{continuous}
Under Assumption \ref{assumption1}, consider the {\color{black}switching} system
\begin{align*}
\frac{dX}{dt}&=f(u[X_n],X,\mu,\nu_n)
\end{align*}with initial conditions $X_i(t_0)\geq0$, $i=1,2,\dots,n$, $K>0$, $k_1>0$, $\mu\geq0$, $\nu_n\geq0$ and the input $u[X_n]$ defined in \eqref{input}. 

Suppose that there are positive values $Y_1^*,Y_2^*,\dots,Y_{n-1}^*$ such that 
\begin{align*}
0&=f_1(Y_1^*,Y_2^*)-f_2(Y_2^*,Y_3^*)-\mu\cdot Y_2^*\\\notag
0&=f_2(Y_2^*,Y_3^*)-f_3(Y_3^*,Y_4^*)-\mu\cdot Y_3^*\\\notag
\vdots\\\notag
0&=f_{n-1}(Y_{n-1}^*,K)-\nu_nf_n(K)-\mu\cdot K\notag,
\end{align*}and define
\begin{align*}
 \alpha:=\frac{1}{k_1}\Big(\sum_{i=1}^{n-1}\mu\cdot Y_i^* + \mu\cdot K + \nu_nf_n(K)\Big).
\end{align*}

For every $m\in\mathbb N$, let $\psi^\varepsilon(t):=(\psi^\varepsilon_1,\psi^\varepsilon_2,\dots,\psi^\varepsilon_n)$ be the solution for the continuous differential equation
\begin{align*}
\frac{d\psi^\varepsilon}{dt}&=f\begin{pmatrix}U_\varepsilon\big[\psi^\varepsilon_n\big], \psi^\varepsilon,\mu,\nu \end{pmatrix}
\end{align*}where 
\begin{align*}
U_\varepsilon\big[\psi^\varepsilon_n\big]:=\begin{cases}
k_1&\text{ if }\psi^\varepsilon_n<K-\varepsilon\\
\frac{k_1(\alpha-1)}{\varepsilon}(\psi^\varepsilon_n-K)+\alpha\cdot k_1&\text{ if }K-\varepsilon\leq\psi^\varepsilon_n<K\\
\frac{\alpha\cdot k_1}{\varepsilon}(K+\varepsilon-\psi^\varepsilon_n)&\text{ if }K\leq\psi^\varepsilon_n<K+\varepsilon\\
0&\text{ if }\psi^\varepsilon_n>K+\varepsilon\\
\end{cases}
\end{align*}$\varepsilon$ is a small number and the initial conditions $\psi^\varepsilon_i(t_0)=X_i(t_0)$ for all $i=1,2,\dots,n$.

{\color{black}Then}, for a.e. $t\in[t_0,\infty)$,
\begin{align*}
\lim_{\varepsilon\to0}\psi^\varepsilon_i(t)&=X_i(t)&\forall i=1,2,\dots,n.
\end{align*}
\end{proposition}

\begin{proof}
Notice that
\begin{align*}
\lim_{\varepsilon\to0}U_\varepsilon\big[\psi_n^\varepsilon\big]=\begin{cases}
k_1&\text{if } \psi_n^\varepsilon<K\\
\alpha\cdot k_1&\text{if }\psi_n^\varepsilon=K\\
0&\text{if }\psi_n^\varepsilon>K
\end{cases}.
\end{align*}
The result follows from Lemma 3 in Section 7, Chapter 2 (p. 82) of \cite{filippov1988differential}.
\end{proof}

\begin{example}
Consider the {\color{black}switching} system \eqref{rMM} with reversible Michaelis-Menten kinetics introduced in \Cref{section_rMM}. \Cref{pslex1}, \Cref{pslex2} and \Cref{pslex3} illustrate the result of \Cref{continuous}.

\begin{figure}[!]\label{pslex1}
\includegraphics[width=\columnwidth]{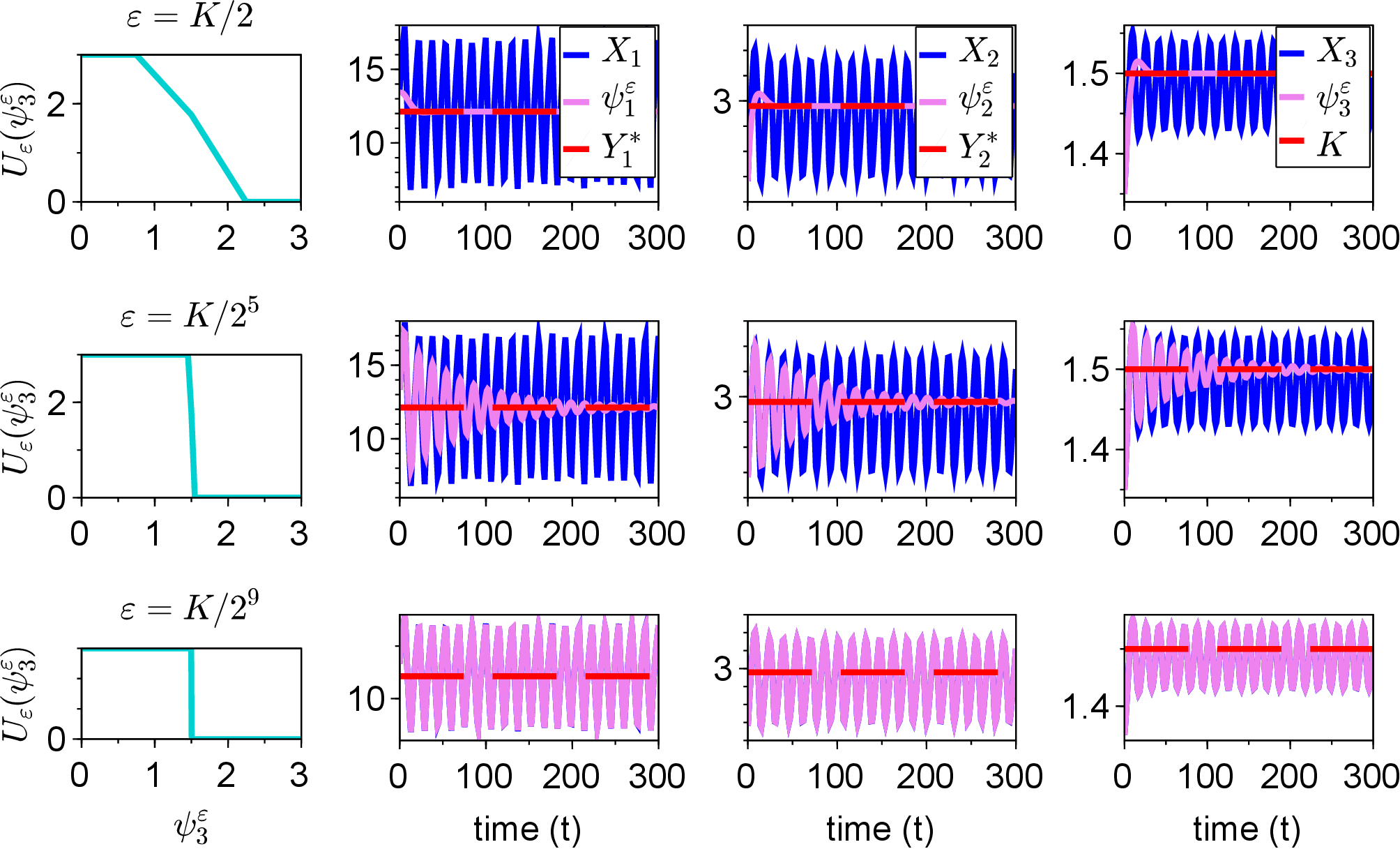}
\caption{The {\color{black}switching} system \eqref{rMM} oscillates around the sliding mode $(Y_1^*, Y_2^*,K)$. Functions $\psi^\varepsilon$ approximate $X$ as $\varepsilon\to0$. All parameters are equal to $1$ except for $k_1=3$, $\nu_n=0.2$, $\mu=0.1$ and $K=1.5$. The initial conditions are $X_1(0)=13.34$, $X_2(0)=2.68$ and $X_3(0)=1.35$. }
\end{figure} 

\begin{figure}[!]\label{pslex2}
\includegraphics[width=\columnwidth]{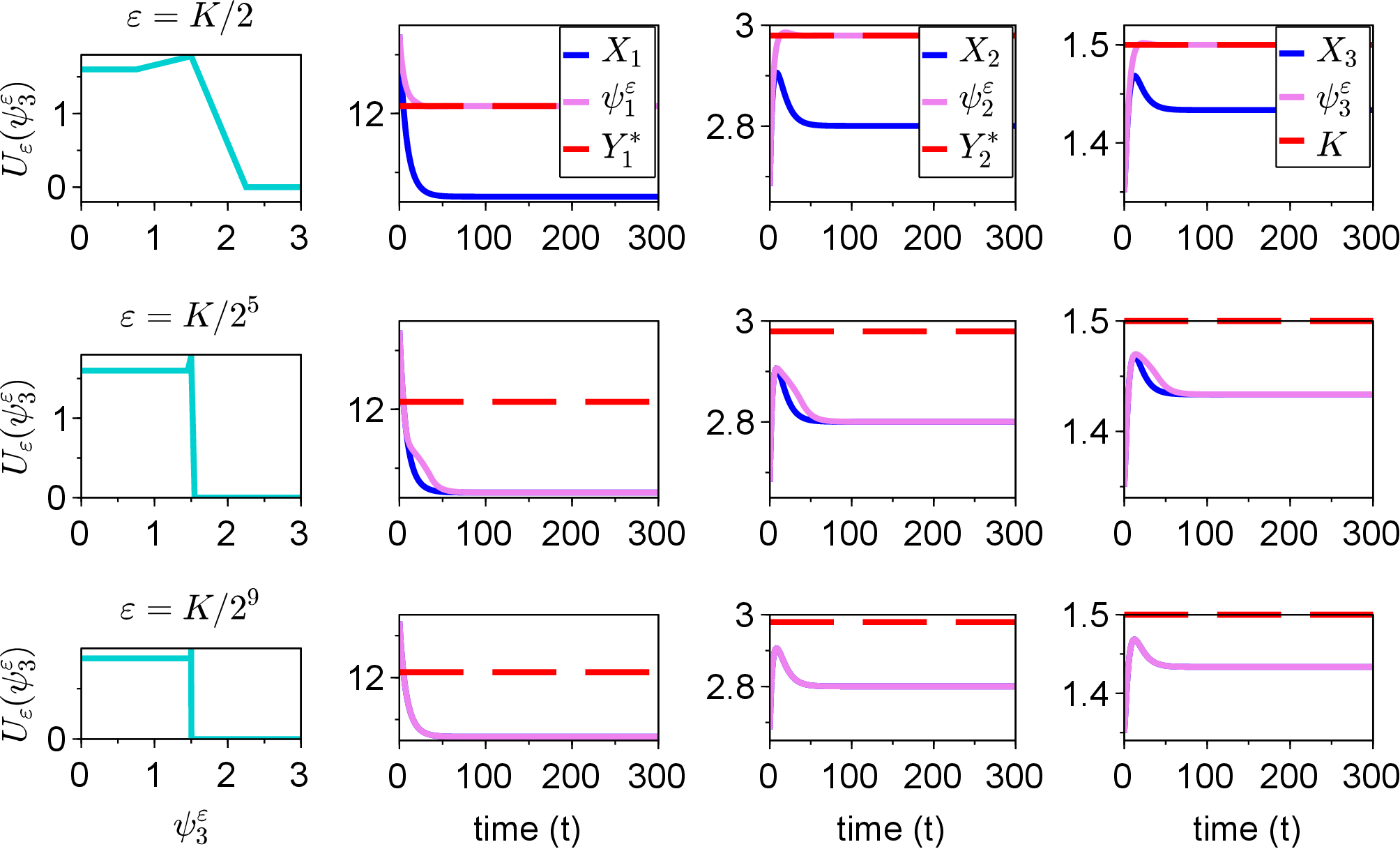}
\caption{The {\color{black}switching} system \eqref{rMM} converges to an equilibrium point {\color{black}smaller than} the sliding mode $(Y_1^*, Y_2^*,K)$. Functions $\psi^\varepsilon$ approximate $X$ as $\varepsilon\to0$. All parameters are equal to $1$ except for $k_1=1.6$, $\nu_n=0.2$, $\mu=0.1$ and $K=1.5$. The initial conditions are  $X_1(0)=13.34$, $X_2(0)=2.68$ and $X_3(0)=1.35$. }
\end{figure} 

\begin{figure}[!]\label{pslex3}
\includegraphics[width=\columnwidth]{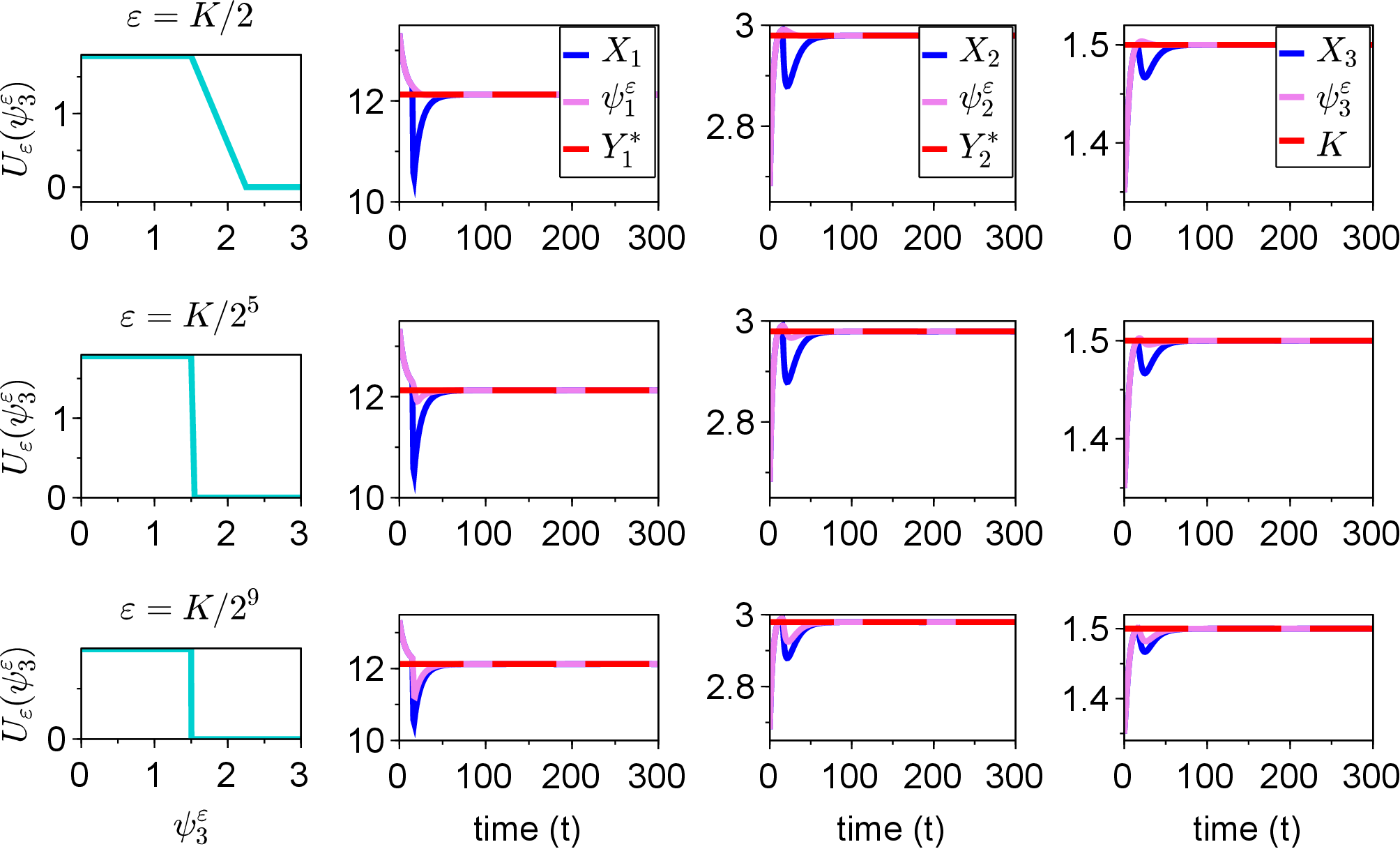}
\caption{The {\color{black}switching} system \eqref{rMM} converges to the sliding mode $(Y_1^*, Y_2^*,K)$. Functions $\psi^\varepsilon$ approximate $X$ as $\varepsilon\to0$. All parameters are equal to $1$ except for $k_1=\sum_{i=1}^{n-1}\mu\cdot Y_i^* + \mu\cdot K + \nu_nf_n(K)\sim 1.78$, $\nu_n=0.2$, $\mu=0.1$ and $K=1.5$. The initial conditions are $X_1(0)=13.34$, $X_2(0)=2.68$ and $X_3(0)=1.35$. }
\end{figure} 
\end{example}

\newpage

\section{Discussion and Conclusions}
In this work, a model to represent allosteric regulation in a metabolic pathway is presented. In this approach, the enzymes are considered to have slow dynamics and to be therefore constant. Metabolites have faster dynamics and a ODE system based on an end-product control structure \cite{goelzer2008reconstruction, goelzer2014towards} is studied. 

Considering that allosteric processes occur very fast, the mechanis{\color{black}m} of regulation is supposed to act as a switched feedback control that is modulated according to the concentration of the end-product of the metabolic pathway. Then, a differential inclusion is defined for the discontinuous {\color{black}switching} system and the existence and right uniqueness of an absolutely continuous solution is proved.

Moreover, the qualitative behavior of the absolutely continuous solution is analyzed and three possible trajectories are observed:
\begin{itemize}
\item There is a positive sliding mode and the {\color{black}switching} system reaches it. Then the {\color{black}switching} system stabilizes at the sliding mode.
\item There is a positive sliding mode and the constant input ON is larger than a threshold defined by the sliding mode. Then the {\color{black}switching} system oscillates around the sliding mode.
\item The system with the constant input ON has a positive equilibrium point and the {\color{black}switching} system converges uniformly and asymptotically to this, because the constant input ON is smaller than the threshold defined by the positive sliding mode (in this case the equilibrium point of the {\color{black}switching} system is equal to or lower than the sliding mode entry by entry), or because there is no positive sliding mode.  
\end{itemize}

{\color{black}Novel features presented in this work for the study of oscillatory behaviors are the reversibility of reactions in the metabolic pathway and the use of differential inclusions to characterize the system solutions.

The approach of Filippov \cite{filippov1988differential} was necessary to study the discontinuous systems \eqref{end_product_feedback}, \eqref{end_product_feedback_st} and \eqref{diff_decay}. The existence and uniqueness of their solutions are based on that. In \Cref{sec_continuous}, we have proved that the dynamics of the discontinuous system can be approximated by continuous equations with more complex inputs than the classical sigmoid or monotone piecewise linear functions. We can then conclude that the analysis with differential inclusions has allowed to also rigorously characterize the dynamical behavior of a class of continuous feedback systems with reversible reactions and smooth or piecewise linear inputs. These results are new to our knowledge, specially concerning the oscillatory behavior.}


Finally, the results obtained with this approach are potentially useful for reducing a genetic-metabolic network with slow and fast dynamics using the theory of singularly perturbed system.

\appendix
\section{Proof of \Cref{lemma1} when $\mu=0$ and $\nu_n>0$}\label{proof_gouze}
\begin{proof}
Define the (Lyapunov) norm-like function $$V(X):=\sum_{i=1}^n\vert X_i-X_i^*\vert.$$ $V$ is nonnegative and $V(X)=0$ if and only if $X_i=X_i^*$ for every $i=1,\dots,n$. Moreover,
\begin{align*}
\dot V=&\sum_{i=1}^n \dot X_i\cdot sgn(X_i-X_i^*).
\end{align*}

On the other hand, the existence of the equilibrium point guarantees 
\begin{align*}
f_i(X_i^*,X_{i+1}^*)&=\mathbf I & \forall i=1,\dots,n-1,\\
\nu_nf_n(X_n^*)&=\mathbf I.
\end{align*}

Then, after some algebraic computations we obtain
\begin{align*}
\dot V\leq&0
\end{align*}Therefore, $V(X)$ is decreasing. This implies that the trajectories of $X$ are bounded, since the distance to the nonnegative equilibrium point $X^*$ is nonincreasing.

On the other hand, the Jacobian of \eqref{input_on_I}
$$J(X)=\begin{pmatrix}
-\frac{\partial f_1}{\partial X_1} & -\frac{\partial f_1}{\partial X_2}& \dots & 0 \\
\frac{\partial f_1}{\partial X_1} & \frac{\partial f_1}{\partial X_2}-\frac{\partial f_2}{\partial X_2}  & \dots &0 \\
0 & \frac{\partial f_2}{\partial X_2}  & \dots & 0 \\
\vdots & \vdots & \ddots & \vdots \\
0 & 0 & \dots & -\frac{\partial f_{n-1}}{\partial X_n}\\
0 & 0 & \dots & \frac{\partial f_{n-1}}{\partial X_n}-\nu_n\frac{\partial f_n}{\partial X_n}
\end{pmatrix}$$is a compartmental matrix for every $X$ thanks to the monotonicity of the functions $f_i$ established Assumption \ref{assumption1}. By Theorem 5 in \cite{jacquez1993qualitative}, every orbit of \eqref{input_on_K} tends to the equilibrium set, i.e. the equilibrium $X^*$ is globally attractive.
 
Moreover, since $f_n$ is strictly increasing w.r.t. $X_n$, $\nu_n\frac{\partial f_n}{\partial X_n}\not=0$. Then, the Jacobian $J$ is out-flow connected. Hence, by Theorem 3 in \cite{jacquez1993qualitative}, $J$ is nonsingular. By the Gershgorin Disc Theorem, a column diagonally dominant nonsingular matrix is stable. Therefore, $X^*$ is locally asymptotically stable. But, since $X^*$ is also globally attractive, we conclude that $X^*$ is globally asymptotically stable. Finally, since system \eqref{input_on_I} is autonomous, $X^*$ is globally uniformly asymptotically stable.

\end{proof}


\section{Proofs of Lemmas}\label{supp_mat}
\begin{proof}[Proof of \Cref{lemma4}]
Using the monotonicity conditions set in Assumption \ref{assumption1} for the functions $f_i$, it can be proved that $\Omega_j$ is invariant showing that $X(t)$ points into $\Omega_j$ on the boundary of $\Omega_j$ for every $j=1,2,3,4$.
\end{proof}

\begin{proof}[Proof of \Cref{lemma2}]
We will first prove that $Y_n^*<X_n^*$ by contradiction. Suppose that $X_n^*\leq Y_n^*$. We will prove by induction that $X_i^*\leq Y_i^*$ for all $i=1,2,\dots,n-1$ as well.

First, since $f_n$ is strictly increasing and $0\leq\mu$, this implies
\begin{align*}
&f_{n-1}(X_{n-1}^*,X_n^*)=\nu_nf_n(X_n^*)+\mu\cdot X_n^*\\
\leq&\nu_nf_n(Y_n^*)+\mu\cdot Y_n^*=f_{n-1}(Y_{n-1}^*,Y_n^*).
\end{align*}But $f_{n-1}$ is decreasing w.r.t. to the second entry and $X_n^*\leq Y_n^*$, then
\begin{align*}
f_{n-1}(X_{n-1}^*,X_n^*)\leq f_{n-1}(Y_{n-1}^*,Y_n^*)\leq f_{n-1}(Y_{n-1}^*,X_n^*).
\end{align*}Hence, since $f_{n-1}$ is strictly increasing w.r.t. to the first entry,
$$X_{n-1}^*\leq Y_{n-1}^*.$$
The induction hypothesis is that for some $m<n-1$ it is satisfied
\begin{align*}
X_i^*&\leq Y_i^*&\forall i>m.
\end{align*}We will prove that $X_m^*\leq Y_m^*$. The existence of the equilibrium points guarantees that
\begin{align*}
f_m(X_m^*,X_{m+1}^*)&=f_{m+1}(X_{m+1}^*,X_{m+2}^*)+\mu\cdot X_{m+1}^*\\
&=\nu_nf_n(X_n^*)+\sum_{i=m+1}^n\mu\cdot X_i^*,
\end{align*}and
\begin{align*}
f_m(Y_m^*,Y_{m+1}^*)&=f_{m+1}(Y_{m+1}^*,Y_{m+2}^*)+\mu\cdot Y_{m+1}^*\\
&=\nu_nf_n(Y_n^*)+\sum_{i=m+1}^n\mu\cdot Y_i^*,
\end{align*}
But, by the induction hypothesis,
\begin{align*}
\nu_nf_n(X_n^*)+\sum_{i=m+1}^n\mu\cdot X_i^*\leq \nu_nf_n(Y_n^*)+\sum_{i=m+1}^n\mu\cdot Y_i^*.
\end{align*}Thus, 
\begin{align*}
f_m(X_m^*,X_{m+1}^*)&\leq f_m(Y_m^*,Y_{m+1}^*)\leq f_m(Y_m^*,X_{m+1}^*),
\end{align*}where the second inequality is due to hypothesis of induction $X_{m+1}^*\leq Y_{m+1}^*$ and to that $f_m$ is decreasing w.r.t. the second entry. Moreover, since $f_m$ is strictly increasing w.r.t. the first entry, we conclude that $$X_m^*\leq Y_m^*.$$ 

Therefore, we have proved by induction that assuming $X_n^*\leq Y_n^*$ implies
\begin{align*}
X_i^*&\leq Y_i^*&\forall i=1,\dots,n-1.
\end{align*}But, since $f_n$ is strictly increasing, $0\leq\nu_n$ and $0\leq\mu$, this leads to conclude
\begin{align*}
\mathbf I_1=\sum_{i=1}^n\mu\cdot X_i^*+\nu_nf_n(X_n^*)\leq\sum_{i=1}^n\mu\cdot Y_i^*+\nu_nf_n(Y_n^*)=\mathbf I_2,
\end{align*}which contradicts the hypothesis of that $\mathbf I_2<\mathbf I_1$. We conclude that $Y_n^*<X_n^*$.
With similar arguments as above, we prove that $Y_i^*<X_i^*$ for every $i=1,2,\dots,n-1$. 
\end{proof}

\begin{proof}[Proof of \Cref{lemma9}]
We have that 
\begin{align*}
f_{n-1}(X_{n-1}^*,X_n^*)=\nu_nf_n(X_n^*)+\mu\cdot X_n^*.
\end{align*}Then, since $f_{n-1}$ is continuous and strictly increasing w.r.t. the first entry,
\begin{align*}
\nu_nf_n(X_n^*)+\mu\cdot X_n^*<\lim_{X_{n-1}\to \infty}f_{n-1}(X_{n-1},X_n^*).
\end{align*}By the continuity of $f_n$ and $f_{n-1}$, it follows that
\begin{align*}
&\nu_nf_n(X_n^*)+\mu\cdot X_n^*=\lim_{X_n\to X_n^*}\nu_nf_n(X_n)+\mu\cdot X_n\\
<&\lim_{X_{n-1}\to \infty}f_{n-1}(X_{n-1},X_n^*)=\lim_{X_n\to X_n^*}\lim_{X_{n-1}\to \infty}f_{n-1}(X_{n-1},X_n).
\end{align*}

Then, there exists $X_n^*<X_n'$ such that 
\begin{align*}
&\nu_nf_n(X_n^*)+\mu\cdot X_n^*<\nu_nf_n(X_n')+\mu\cdot X_n'<\lim_{X_{n-1}\to \infty}f_{n-1}(X_{n-1},X_n').
\end{align*}

By induction it can be proved that for any sequence $\{X_n^{\alpha(j)}\}_{j\in\mathbb N}$ such that $X_n^{\alpha(j)}\in(X_n^*,X']$ for every $j=1,2,\dots$ and $$\lim_{j\to\infty}X_n^{\alpha(j)}=X_n^*,$$there exist{\color{black}s} $X_i^*<X_i^{\alpha(j)}$, $i=1,2,\dots,n-1$, such that 
\begin{align*}
0&=f_1(X_1^{\alpha(j)},X_2^{\alpha(j)})-f_2(X_2^{\alpha(j)},X_3^{\alpha(j)})-\mu\cdot X_2^{\alpha(j)}\\
0&=f_2(X_2^{\alpha(j)},X_3^{\alpha(j)})-f_3(X_3^{\alpha(j)},X_4^{\alpha(j)})-\mu\cdot X_3^{\alpha(j)}\\
\vdots\\
0&=f_{n-1}(X_{n-1}^{\alpha(j)},X_n^{\alpha(j)})-\nu_nf_n(X_n^{\alpha(j)})-\mu\cdot X_n^{\alpha(j)}.
\end{align*}Define
$$\mathbf I^{\alpha(j)}:=f_1(X_1^{\alpha(j)},X_2^{\alpha(j)})+\mu\cdot X_1^{\alpha(j)}.$$
This {\color{black}is satisfied}, because $f_n$ is strictly increasing,
\begin{align*}
\mathbf I_1&=f_n(X_n^*)+\sum_{i=1}^n\mu\cdot X_i^*\\
&<f_n(X_n^{\alpha(j)})+\sum_{i=1}^n\mu\cdot X_i^{\alpha(j)}=\mathbf I^{\alpha(j)},
\end{align*}and
$$\lim_{j\to\infty}\mathbf I^{\alpha(j)}=\mathbf I_1.$$

We conclude that, for any $0<\varepsilon$, there exists $\mathbf I_1<\mathbf I^{\alpha(j')}<\mathbf I_1+\varepsilon$ such that $X_i^*<X_i^{\alpha(j')}$ for all $i=1,\dots,n$ and
\begin{align*}
0&=\mathbf I^{\alpha(j')}-f_1( X_1^{\alpha(j')}, X^{\alpha(j')}_2)-\mu\cdot  X^{\alpha(j')}_1\\
0&=f_1( X^{\alpha(j')}_1, X^{\alpha(j')}_2)-f_2( X^{\alpha(j')}_2, X^{\alpha(j')}_3)-\mu\cdot  X^{\alpha(j')}_2\\
\vdots\\\notag
0&=f_{n-1}( X^{\alpha(j')}_{n-1}, X^{\alpha(j')}_n)-\nu_nf_n( X^{\alpha(j')}_n)-\mu\cdot  X^{\alpha(j')}_n.
\end{align*}

\end{proof}

\begin{proof}[Proof of \Cref{lemma5}]
Define
\begin{align*}
Z_i(t)&:=X_i(t)-Y_i(t)& t\geq t_0, \forall i=1,2,\dots,n.
\end{align*}
The first conclusion of the Lemma follows from the fact that the sets
$$\mathbb R^n_+:=\underset{n\text{-times}}{[0,\infty)\times[0,\infty)\times\dots\times[0,\infty)},$$
$$int(\mathbb R^n_+):=\underset{n\text{-times}}{(0,\infty)\times(0,\infty)\times\dots\times(0,\infty)},$$are invariant under the flow $Z(t)$ if $\mathbf I_2<\mathbf I_1$.

On the other hand, if $\mathbf I_1=\mathbf I_2$ and $Y_i(t_0)<X_i(t_0)$ for every $i=1,2, \dots,n$, then, by the continuity of $Z$, there is $\varepsilon>0$ such that for all $i=1,2,\dots,n$
\begin{align*}
0<Z_i(t)&:=X_i(t)-Y_i(t)& t_0<t<t_0+\varepsilon.
\end{align*}With similar arguments as above, it can be proved that $\mathbb R^n_+$ is positively invariant under the flow of $Z$. 
\end{proof}

\begin{proof}[Proof of \Cref{lemma11}]
The proof can be done by induction over $m$, i.e., first proving the Lemma when $m=n-1$, assuming the Lemma as the induction hypothesis for some $m+1$ and proving it for the case $m$.
\end{proof}

\begin{proof}[Proof of \Cref{lemma12}]
By hypothesis,
\begin{align*}
sgn(\dot{\widetilde X}_n(t_0))&=sgn(\dot Z_n(t_0)).
\end{align*}
Suppose that $sgn(\dot{\widetilde X}_n(t_0))<0$. Then, by the continuity of $\dot{\widetilde X}_n$ and $\dot Z_n$, there exist $\varepsilon_1$ and $\varepsilon_2$ such that
\begin{align*}
sgn(\dot{\widetilde X}_n(t))&<0&\forall t\in(t_0,t_0+\varepsilon_1),\\
sgn(\dot Z_n(t))&<0&\forall t\in(t_0,t_0+\varepsilon_2).
\end{align*}Hence, if $\varepsilon:=\min\{\varepsilon_1,\varepsilon_2\}$,
\begin{align*}
sgn(\dot{\widetilde X}_n(t))&=sgn(\dot Z_n(t))<0&\forall t\in(t_0,t_0+\varepsilon).
\end{align*}

Analogously, if $0<sgn(\dot{\widetilde X}_n(t_0))$, by the continuity of $\dot{\widetilde X}_n$ and $\dot Z_n$, there exist $\varepsilon_1$ and $\varepsilon_2$ such that
\begin{align*}
0<&sgn(\dot{\widetilde X}_n(t))&\forall t\in(t_0,t_0+\varepsilon_1),\\
0<&sgn(\dot Z_n(t))&\forall t\in(t_0,t_0+\varepsilon_2).
\end{align*}Hence, if $\varepsilon:=\min\{\varepsilon_1,\varepsilon_2\}$,
\begin{align*}
0<sgn(\dot{\widetilde X}_n(t))&=sgn(\dot Z_n(t))&\forall t\in(t_0,t_0+\varepsilon).
\end{align*}

Assume that $sgn(\dot{\widetilde X}_n(t_0))=0$. If there exist $\varepsilon_1$ and $\varepsilon_2$ such that
\begin{align*}
sgn(\dot{\widetilde X}_n(t))<0&&\forall t\in(t_0,t_0+\varepsilon_1),\\
0<sgn(\dot Z_n(t))&&\forall t\in(t_0,t_0+\varepsilon_2),
\end{align*}then, 
\begin{align*}
\dot{\widetilde X}_n(t)<\dot{\widetilde X}_n(t_0)<\dot Z_n(t)&&\forall t\in(t_0,t_0+\varepsilon),
\end{align*}where $\varepsilon:=\min\{\varepsilon_1,\varepsilon_2\}$. But the inequality above contradicts \Cref{lemma5}, since $0<k_1$.\\

We cannot suppose there exist $\varepsilon_1$ and $\varepsilon_2$ such that
\begin{align*}
sgn(\dot{\widetilde X}_n(t))&=0&\forall t\in[t_0,t_0+\varepsilon_1)&&\text{or}\\
sgn(\dot Z_n(t))&=0&\forall t\in[t_0,t_0+\varepsilon_2),
\end{align*}because this leads to conclude that the systems are at equilibrium in $[t_0,t_0+\varepsilon)$.

But this contradicts the hypothesis, as it was assume{\color{black}d} 
\begin{align*}
\dot{Z}_m(t_0)&=\dot{\widetilde X}_m(t_0)\not=0&\text{for some } m\in\{1,2,\dots,n\}.
\end{align*}

Now suppose $sgn(\dot{\widetilde X}_n(t_0))=0$ and there exist $\varepsilon_1$ and $\varepsilon_2$ such that
\begin{align*}
0<sgn(\dot{\widetilde X}_n(t))&&\forall t\in(t_0,t_0+\varepsilon_1),\\
sgn(\dot Z_n(t))<0&&\forall t\in(t_0,t_0+\varepsilon_2).
\end{align*}

(One option is to say that this contradicts \Cref{lemma_uniqueness} (existence of solution theorem of Filippov, Theorem 1, p. 77 in \cite{filippov1988differential}), another option is as follows).

According to the hypothesis, $\dot{\widetilde X}_m(t_0)=\dot Z_m(t_0)\not=0$ for some $1<m<n$. Without loss of generality, assume that $\dot{\widetilde X}_i(t_0)=\dot Z_i(t_0)=0$ for every $m<i$ Hence, let us assume that there is $0<\varepsilon<\min\{\varepsilon_1,\varepsilon_2\}$ such that
\begin{align*}
0&<sgn(\dot{\widetilde X}_m(t))=sgn(\dot Z_m(t))&\forall t\in[t_0,t_0+\varepsilon).
\end{align*}

This implies, according to \Cref{lemma11},
\begin{align*}
0&<sgn(\dot{\widetilde X}_i(t))=sgn(\dot Z_i(t))&\forall t\in(t_0,t_0+\varepsilon), \forall i>m,
\end{align*}that contradicts the supposition 
\begin{align*}
sgn(\dot Z_n(t))<0&&\forall t\in(t_0,t_0+\varepsilon_2).
\end{align*}

Analogously, by \Cref{lemma11}, assuming 
\begin{align*}
sgn(\dot{\widetilde X}_m(t))=sgn(\dot Z_m(t))&<0&\forall t\in[t_0,t_0+\varepsilon),
\end{align*}implies for all $i>m$
\begin{align*}
sgn(\dot{\widetilde X}_i(t))=sgn(\dot Z_i(t))<0&&\forall t\in(t_0,t_0+\varepsilon),
\end{align*}that contradicts the supposition 
\begin{align*}
0<sgn(\dot {\widetilde X}_n(t))&&\forall t\in(t_0,t_0+\varepsilon_1).
\end{align*}

We conclude that is false that $sgn(\dot{\widetilde X}_n(t_0))=0$ and there exist $\varepsilon_1$ and $\varepsilon_2$ such that
\begin{align*}
0<sgn(\dot{\widetilde X}_n(t))&&\forall t\in(t_0,t_0+\varepsilon_1),\\
sgn(\dot Z_n(t))<0&&\forall t\in(t_0,t_0+\varepsilon_2).
\end{align*}

\end{proof}

\begin{proof}[Proof of \Cref{lemma10}]
Consider any $X_n'\in(0,X_n^*)$. According to Assumption \ref{assumption1}, 
$$f_{n-1}(0,X_n')\leq0<\nu_nf_n(X_n')+\mu\cdot X_n'< f_{n-1}(X_{n-1}^*,X_n')$$
Then, since $f_{n-1}$ is strictly increasing w.r.t. the first entry, there exists an unique $X_{n-1}'\in(0,X_{n-1}^*)$ such that
\begin{align*}
f_{n-1}(X_{n-1}',X_n')=\nu_nf_n(X_n')+\mu\cdot X_n'.
\end{align*}The proof follows by induction. 
\end{proof}

\begin{proof}[Proof of \Cref{lemma8}]
The trivial case where the hybrid system $X$ never switches follows from \Cref{lemma5}.

Now suppose that the hybrid system switches. By the continuity of $X_n$, there exists a countable $\mathbf N:=\{0,1,2,\dots,k\}\subset\mathbb N$ and values $\{t_j\}_{j\in\bold N}\subset\mathbb R^+$, with $t_j<t_{j+1}$ for every $j\in\bold N$, such that the {\color{black}switching} system is under a single regime in any interval $(t_j,t_{j+1})$ for every $j\in\bold N$ (i.e. $X_n(t)\leq K$ for all $t\in(t_j,t_{j+1})$ or $K<X_n(t)$ for all $t\in(t_j,t_{j+1})$).

It can be proved by induction, using \Cref{lemma5}, that for every $j\in\bold N$,
\begin{align*}
X_i(t)&\leq \widetilde X_i(t)&\forall t\in(t_j,t_{j+1}), \forall i=1,2,\dots,n.
\end{align*}Furthermore, by continuity of $X_i$ and $\widetilde X_i$,
\begin{align*}
X_i(t_j)&\leq\widetilde X_i(t_j)& \forall i=1,2,\dots,n,&\forall j\in\bold N.
\end{align*} 
\end{proof}

{\color{black}\section{Comparison with a piecewise linear model}\label{app_pasternack}
Here we present an example for comparison with the approach of piecewise linear models first proposed by Glass and Pasternack in \cite{glass1978stable}. The following piecewise linear model is based on the model proposed by Poignard et al. in \cite{poignard2016periodic}:
\begin{align}\label{poignard}
\frac{dX_{1,Bool}}{dt}&:=k_1\cdot (1-s^+(X_{3,Bool},K)) -\mu\cdot X_{1,Bool}\\\notag
\frac{dX_{2,Bool}}{dt}&:=k_2\cdot s^+(X_{1,Bool},Y_1^*)-\mu\cdot X_{2,Bool}\\\notag
\frac{dX_{3,Bool}}{dt}&:=k_3\cdot s^+(X_{2,Bool},Y_2^*)-\mu\cdot X_{3,Bool},
\end{align}where the boolean functions are defined as
\begin{align*}
s^+(X_{i,Bool},Y_i^*)&:=\begin{cases}
1&\text{if }X_{i,Bool}>Y_i^*\\
0&\text{if }X_{i,Bool}<Y_i^*
\end{cases}&i=1,2,\\
s^+(X_{3,Bool},K)&:=\begin{cases}
1&\text{if }X_{3,Bool}>K\\
0&\text{if }X_{3,Bool}<K.
\end{cases}
\end{align*}
To compare the piecewise model \eqref{poignard} with a hybrid model as proposed in this article consider
\begin{align}\label{lopez}
\frac{dX_1}{dt}&:= u[X_3] - k_2\cdot X_1 - \mu\cdot X_1\\\notag
\frac{dX_2}{dt}&:= k_2\cdot X_1 - k_3\cdot X_2 - \mu\cdot X_2\\\notag
\frac{dX_3}{dt}&:= k_3\cdot X_2 - \mu\cdot X_3,
\end{align}where $u[X_3]$ is defined as in \eqref{input}. Notice that, in contrast to \eqref{poignard}, mass-action kinetics are considered in \eqref{lopez}. This property of the class of models studied in this work has been recurrently used in the demonstration of our results. 
In \Cref{fig_compare} are depicted numerical solutions for systems \eqref{poignard} and \eqref{lopez}. They both exhibit oscillations crossing the sliding mode $(Y_1^*, Y_2^*, K)$, but with different magnitude and amplitude of oscillations.

\begin{figure}[H]\label{fig_compare}\centering
\includegraphics[width=12cm]{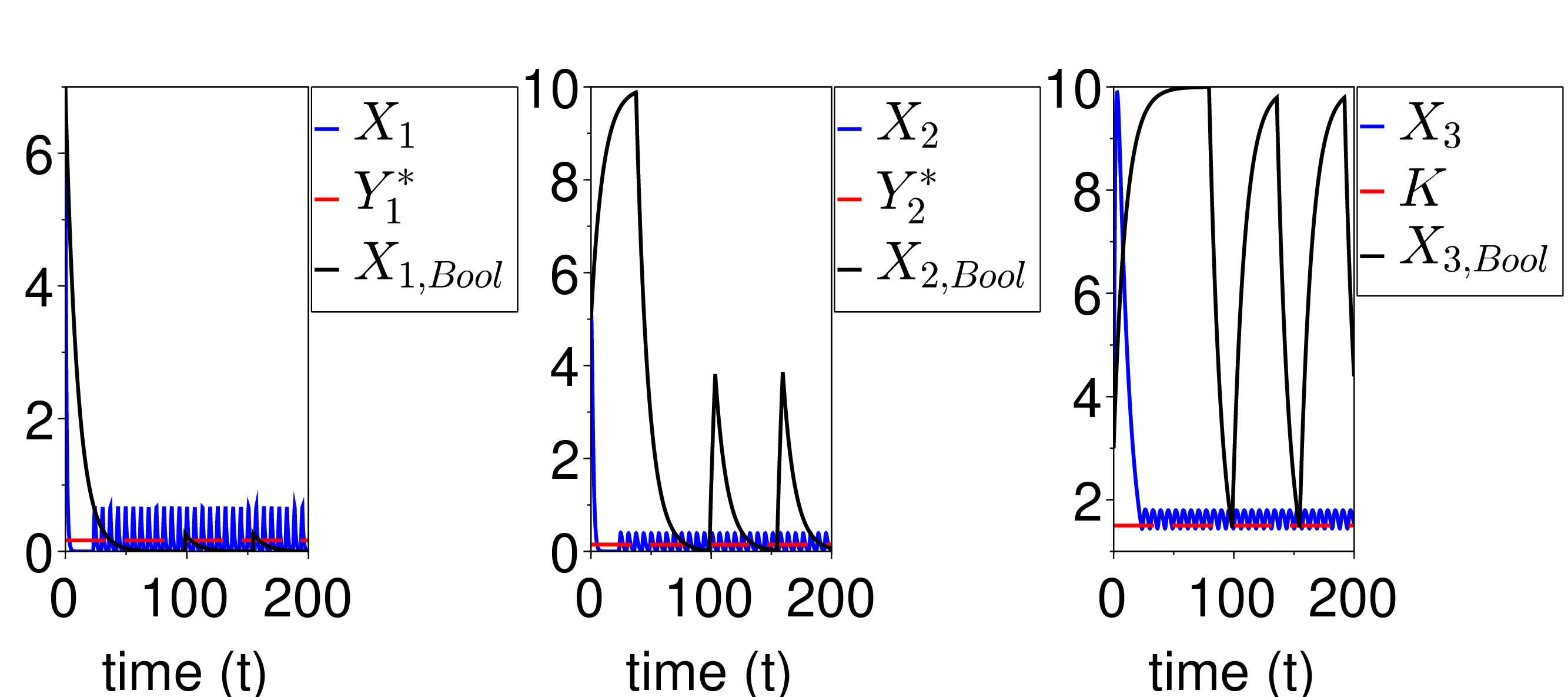}
\caption{\color{black}Numerical solution of \eqref{poignard} and \eqref{lopez}. All parameters are equal to 1 except for $K=1.5$ and $\mu=0.1$. The vector $(Y_1^*, Y_2^*, K)$ corresponds to the sliding mode equilibrium point of \eqref{lopez} (see \Cref{proposition1}). The solutions of \eqref{poignard} and \eqref{lopez} were computed in Scilab using the Runge-Kutta and \emph{ode} functions, respectively. Both models \eqref{poignard} and \eqref{lopez} exhibit oscillations crossing the sliding mode $(Y_1^*, Y_2^*, K)$, but with different magnitude and amplitude of oscillations.}
\end{figure}
}

\section*{Acknowledgments}
Claudia Lopez-Zazueta acknowledges the support of Labex Mathématique Hadamard (LMH).  

\bibliographystyle{siamplain}
\bibliography{mybibfile}
\end{document}